\documentclass[aps,pra,11pt,onecolumn,nofootinbib,tightenlines,superscriptaddress]{revtex4-2}

\usepackage[export]{adjustbox}
\usepackage{tikz}
\usepackage{mathrsfs}
\usepackage{graphicx}
\usepackage{comment}
\usepackage{dsfont}
\usepackage{amsmath}
\usepackage{amsfonts}
\usepackage{amssymb}
\usepackage{braket}
\usepackage{amsthm}
\usepackage{algorithm}
\usepackage{algpseudocode}

\usepackage[colorlinks=true,urlcolor=blue,citecolor=blue,linkcolor=blue]{hyperref}
\usepackage{caption}
\usepackage{lipsum}
\usepackage{subcaption}
\usepackage{mwe}
\usepackage{enumerate}
\usepackage{sidecap}
\usepackage{natbib}
\usepackage{appendix}
\usepackage{multirow}
\usepackage{array}
\usepackage{physics}
\usepackage{mathtools}
\usepackage{empheq}
\usepackage{booktabs} 
\usepackage{siunitx}

\usepackage[paperwidth=215mm,paperheight=297mm,centering,hmargin=2.45cm,vmargin=2.5cm]{geometry}

\newtheorem{theorem}{Theorem}
\newtheorem{lemma}[theorem]{Lemma}
\newtheorem{proposition}[theorem]{Proposition}
\newtheorem{corollary}[theorem]{Corollary}

\theoremstyle{remark}
\newtheorem{remark}{Remark}

\newcommand{\widebar}[1]{\overline{\mkern-4mu#1\mkern-1mu}}

\newcommand{\state}{\mathcal{S}(A)}

\DeclareMathOperator*{\argmin}{arg\,min}
\DeclareMathOperator*{\argmax}{arg\,max}

\renewcommand{\tilde}[1]{\widetilde{#1}} 
\renewcommand{\hat}[1]{\widehat{#1}}

\def \d {\mathrm{d}}

\DeclareMathOperator{\supp}{supp}
\DeclareMathOperator{\Trm}{Tr}

\makeatletter

\renewcommand*{\thesubsection}{\thesection.\arabic{subsection}}
\renewcommand*{\p@subsection}{}

\renewcommand*{\p@subsubsection}{}
\makeatother

\begin{document}

\title{A fixed-point algorithm for matrix projections with applications in quantum information}

\begin{abstract}
 We develop a fixed-point iterative algorithm that computes the matrix projection with respect to the Bures distance on the set of positive definite matrices that are invariant under some symmetry. We prove that the fixed-point iteration algorithm converges exponentially fast to the optimal solution in the number of iterations.  Moreover, it numerically shows fast convergence compared to the off-the-shelf semidefinite program solvers. Our algorithm, for the specific case of Bures-Wasserstein barycenter, recovers the fixed-point iterative algorithm originally introduced in ({\'A}lvarez-Esteban et al., 2016). Our proof is concise and relies solely on matrix inequalities. Finally, we discuss several applications of our algorithm in quantum resource theories and quantum Shannon theory.
\end{abstract}

\author{Shrigyan Brahmachari}
\affiliation{Centre for Quantum Technologies, National University of Singapore, Singapore 117543, Singapore}
\affiliation{Department of Electrical and Computer Engineering, Duke University, Durham, NC 27708, USA}

\author{Roberto Rubboli}
\email{roberto.rubboli@u.nus.edu}
\affiliation{Centre for Quantum Technologies, National University of Singapore, Singapore 117543, Singapore}

\author{Marco Tomamichel}
\affiliation{Centre for Quantum Technologies, National University of Singapore, Singapore 117543, Singapore}
\affiliation{Department of Electrical and Computer Engineering,
National University of Singapore, Singapore 117583, Singapore}

\maketitle

\section{Introduction} 
Given a positive matrix $R$ and a group $G$ with a unitary representation \(\{U_g\}_{g \in G}\), consider the optimization problem  
\begin{align}  
\label{problem_introduction}  
\argmin_{S \geq 0} B(R, \mathcal{E}(S)) \,,
\end{align}
where $B$ denotes the Bures distance between matrices and $\mathcal{E}$ is the group avaraging defined as
\begin{equation}
     \mathcal{E}(S)=\frac{1}{|G|}\sum_{g \in  G}U_g S U^{\dagger}_g \,.
\end{equation}
This is the problem of finding the closest positive matrix to $R$ that is invariant under the group action, as measured by the Bures distance. 
This problem is of interest as it encompasses several important quantities in quantum information theory (see  Section~\ref{quantum info} for a detailed discussion of these problems). Moreover, the problem includes as a special case
the Bures-Wasserstein barycenter,
which seeks to find the barycenter of a set of positive matrices $\{X_j\}_{j=1}^m$ with corresponding weights $\{\omega_j\}_{j=1}^m$ (see Section~\ref{remark barycenter} for details). Explicitly, the Bures-Wasserstein barycenter is the solution of 
\begin{equation}
\label{First Bures-Wasserstein barycenter}
    \argmin_{S \geq 0} \sum_{j=1}^m \omega_j \, B(X_j, S)^2 \,.
\end{equation}
This is particularly relevant in multi-marginal optimal transport
since the 2-Wasserstein distance between two Gaussian measures corresponds to the Bures distance between their covariance matrices (see, e.g., \cite{bhatia2019bures}).
Whereas the problem in~\eqref{problem_introduction} of finding the closest positive invariant matrix has not been studied extensively before, the special case of finding the Bures-Wasserstein barycenter in~\eqref{First Bures-Wasserstein barycenter} has been addressed in several works. 

A fast algorithm to compute the Bures-Wasserstein barycenter problem was introduced in~\cite{alvarez2016fixed}, where the authors proposed a fixed-point iterative algorithm that converges rapidly in practice and provided an asymptotic proof of its convergence. A more concise asymptotic proof was provided in~\cite{bhatia2019bures}, while~\cite{chewi2020gradient} proved a dimension-dependent convergence that is exponentially fast in the number of iterations. Notably, the latter work also provided a geometric interpretation, showing that the fixed-point algorithm for Bures-Wasserstein barycenters can be viewed as a gradient descent on a Riemannian manifold with unit step size.  Building on this interpretation,~\cite{altschuler2021averaging_2} later established dimension-independent theoretical convergence guarantees for the same Riemannian descent algorithm, but with a non-unit step size.\footnote{Note that the argument in the earlier published version~\cite{altschuler2021averaging_1} contained a gap, which has been resolved in the updated arXiv version~\cite{altschuler2021averaging_2}. In particular, at the end of Section 3.1 of~\cite{altschuler2021averaging_2} it is stated that, although a unit step size can be used in practice, a smaller step size is required for the theoretical analysis.}

The Bures–Wasserstein barycenter problem~\eqref{First Bures-Wasserstein barycenter} can also be solved using standard projected gradient descent, which is guaranteed to converge exponentially fast to the optimal solution in terms of iteration complexity. This follows from the fact that the Bures distance is both smooth and strongly convex over the set of positive definite matrices~\cite[Section 4]{bhatia2018strong}. This approach has been followed in~\cite{bhatia2018strong} and~\cite[Appendix C.2]{altschuler2021averaging_2}). However, these algorithms exhibit slower convergence in practice than the fixed-point algorithm (see Section~\ref{Other methods} for a comparison). In particular, their step size is very small when the condition number is large, which makes convergence to the solution slow for large matrices.
Alternatively, the problem~\eqref{First Bures-Wasserstein barycenter} can be solved using standard semidefinite program (SDP) solvers as the Bures distance can be expressed in terms of the fidelity, which admits a semidefinite programming representation~\cite{watrous2012simpler}. However, this method also shows slower convergence in practice (see Section~\ref{Other methods} for a comparison). 
Furthermore, the Riemannian descent with non-unit step sizes proposed in~\cite{altschuler2021averaging_2} performs significantly worse than the fixed-point algorithm, which coincides with Riemannian descent when the step size is set to one (see Section~\ref{Other methods} and Table~\ref{tab:alg-comparison}). This shows that the fixed-point algorithm, or equivalently Riemannian descent with unit step size, offers a fast approach for computing the Bures–Wasserstein barycenter.
Hence, the above discussion raises the question of whether a broader class of problems could be solved more efficiently using a similar fixed-point algorithm.

In our work, we propose a fixed-point algorithm for the closest invaraint matrix as described in equation~\eqref{problem_introduction}, extending the fixed-point iterative algorithm for Bures-Wasserstein barycenters. Our algorithm can alternatively be viewed as a Riemannian descent on the Bures–Wasserstein manifold with unit step size~\cite[Section 3.6.2.]{afham2025thesis}, providing a geometric interpretation. Within this framework, using a uniform step size of one enables the algorithm to converge rapidly to the solution.
We provide a proof of convergence with problem-dependent theoretical convergence guarantees for the most general case, while also offering a dimension-independent proof for the case of commuting invariant matrices, which include quantities of interest in quantum information (see Section~\ref{quantum info}).
Our proof is concise, relies solely on matrix inequalities, and, for the specific Bures-Wasserstein barycenter, for which no dimension-independent proof is known for unit step size, it provides an alternative to the dimension-dependent proof in~\cite{chewi2020gradient}.\footnote{In~\cite{altschuler2021averaging_2}, the authors provide dimension-independent theoretical convergence guarantees for Riemannian descent with non-unit step size for the Bures-Wasserstein problem. Moreover, they comment that, in practice, the step size could be set to one, in which case the algorithm simplifies to the fixed-point method analyzed in this work. Nonetheless, they do not provide a rigorous proof of this fact.}
While our proof is dimension-dependent in the worst-case scenario for the Bures-Wasserstein barycenter, we provide a problem-dependent bound that performs well for well-behaved instances.

 \medskip
 
\textbf{Structure of the manuscript.} The manuscript is organized as follows. Section~\ref{sec: notation} introduces the notation. In Section~\ref{the problem}, we formally present the problem statement, propose our algorithm, and state its convergence guarantees, while also highlighting its connection to the Bures-Wasserstein barycenter problem. Section~\ref{sec: fixed-point equation} derives the fixed-point equation that motivates our algorithmic approach. The proof of our main result is presented in Section~\ref{proof of Theorem 1}, which combines two key inequalities to establish convergence. Section~\ref{quantum info} explores applications in quantum information theory, and Section~\ref{numerical comparison} provides a comprehensive numerical comparison between our fixed-point algorithm and existing methods.

\section{Notation}
\label{sec: notation}
  A Hermitian matrix $R$ is positive semidefinite if all its eigenvalues are nonnegative. We denote with $\mathcal{P}(A)$ the set of positive semidefinite matrices on a Hilbert space $A$, and we often use the notation $R \geq 0$. A Hermitian matrix $R$ is positive definite if all its eigenvalues are positive. In this case, we write $R>0$.  Moreover, we denote with $\state$ the set of quantum states, i.e., the subset of $\mathcal{P}(A)$ with unit trace. The fidelity between two positive semidefinite matrices $R,S \in \mathcal{P}(A)$ is~\cite{uhlmann1985transition}
\begin{equation}
F(R, S) =  \left(\Trm[(R^\frac{1}{2} S R^\frac{1}{2})^\frac{1}{2}]\right)^2 \,.
\end{equation}
The Bures distance is~\cite{bures1969extension}
\begin{equation}
B(R, S) = (\Tr[R]+\Tr[S]-2F(R,S)^\frac{1}{2})^\frac{1}{2} \,.
\end{equation}
Throughout this work, we fix the first argument and we optimize the second argument over the set of positive matrices that is invariant under some symmetry. Hence, below we often use the shorthand $B_R(S):=B(R,S)$.
We denote with $\langle R,S\rangle$ the Hilbert-Schmidt product between two matrices $\langle R,S\rangle = \Tr[R^\dagger S]$. The Frobenius norm is $\|R\|=\sqrt{\langle R,R\rangle}$ and one-norm is $\|R\|_1=\Tr \big[\sqrt{R^\dagger R}\big]$. The fully mixed state is $\pi_d=I/d$ where $d$ is the dimension of the space. We call $\lambda_{\min}(R)$ and $\lambda_{\max}(R)$ the minimum and maximum eigenvalue of a positive semidefinite matrix $R$, respectively. Given Hermitian matrices $X$ and $Y$, we say $X \geq Y$ in the Loewner partial order if $X-Y$ is positive semidefinite. We also often write $A \in [B, C]$ to denote the condition $B \leq A \leq C$. Finally, $[A,B]=AB-BA$ is the commutator of $A$ and $B$.

In the following, we consider the set of positive matrices that are invariant under some symmetry. For an introduction about groups and representations, we refer to the standard textbooks~\cite{raczka1986theory,helgason1979differential}. For a more modern introduction and related topics in quantum information, we refer to~\cite{marvian2012symmetry}. 
Given a group $G$ with projective unitary representation $\{U_g\}_{g \in G}$, we say that a matrix $S$ is invariant if it does change under the action of a projective unitary representation of a group, i.e., $S=U_g S U^{\dagger}_g$ for any $g \in G$.
We denote the set of invariant positive semidefinite matrices and invariant quantum states as $\mathcal{P}_G(A)$ and $\mathcal{S}_G(A)$, respectively. 
Explicitly, 
\begin{align}
    &\mathcal{P}_G(A) = \{ S\in \mathcal{P}(A): S=U_g S U^{\dagger}_g \;\; \text{for all} \;\; g \in G\}\,, \\
    &\mathcal{S}_G(A) = \{ S\in \mathcal{S}(A): S=U_g S U^{\dagger}_g \;\; \text{for all} \;\; g \in G\} \,.\footnotemark
\end{align}
\addtocounter{footnote}{-1} 
\footnotetext{
The set of all invariant matrices forms a unital $^*$-subalgebra of the matrix algebra. Conversely, by a standard application of the double-commutant theorem, any unital \( ^*\)-subalgebra can be realized as the fixed-point algebra under conjugation by the unitaries in its commutant. The set \( \mathcal{P}_G(A) \) is the positive cone of positive semidefinite matrices inside the unital $^*$-subalgebra of invariant matrices. Note that \( \mathcal{P}_G(A) \) is not a unital \( ^*\)-subalgebra of the matrix algebra as it is not closed under multiplication. 
}
For a Hermitian matrix $X$, the group averaging is defined by the following linear map 
\begin{equation}
\mathcal{E}(X)=\frac{1}{|G|}\sum_{g \in  G}U_g X U^{\dagger}_g \,.
\end{equation}

To characterize the optimal solution, we utilize the Fr\'echet derivative to analyze the first-order optimality condition (see~\cite{bhatia1997matrix} for a review).
Given a convex function \( f: \mathcal{P}(A) \to \mathbb{R}_{+} \),  the Fr\'echet derivative of $f$ at $S$ is the linear map from the set of Hermitian matrices into $\mathbb{R}_{+}$, and its action is given by
\begin{equation}
Df(S)(Z)= \frac{\d}{\d x} \bigg|_{x=0} f(S+x Z) \,.
\end{equation}

\section{Fixed-point algorithm for matrix projections with respect to the Bures distance}
\label{the problem}
In this section, we present our main result. We begin by formally introducing the problem, followed by a description of the algorithm. We then state our main theorem, which establishes the convergence of the algorithm with respect to the number of iterations. Finally, we explore the connection to the Bures-Wasserstein barycenter problem.

We begin by defining the central problem of this work: determining the positive invariant matrix that minimizes the Bures distance to a given fixed matrix.

Let $G$ be a group with projective unitary representation $\{U_g\}_{g \in G}$ and $R$ be a positive semidefinite matrix.  Throughout this work, we want to find the positive invariant matrix that minimizes the Bures distance to $R$, i.e, we want to find 
\begin{equation}
\label{first problem}
\argmin_{S \in \mathcal{P}_G(A)} B(R,S)^2 \,.
\end{equation}
We observe that this problem can be equivalently expressed by introducing group averaging and framing it as an optimization problem over positive semidefinite matrices, yielding the form presented in equation~\eqref{problem_introduction} in the introduction.

The positive invariant matrix that solves the above optimization problem gives the projection with respect to the Bures distance of $R$ onto the set of positive invariant matrices.

\medskip

The fixed-point algorithm is given as follows:
\begin{algorithm}[H] 
\normalsize
\caption{Fixed-point iterative algorithm}\label{Fixed-point algo}
\label{alg:loop}
\begin{algorithmic}[1]
\Require{Positive definite matrix $R$, number of iterations $N$} 
\State {\textbf{Initialize:} {Initial point $S_0=$ {$\mathcal{E}(R^\frac{1}{2})^2$}}}
\For{$n = 1,...,N$}                    
    \State {$S_{n} = $  {$S_{n-1}^{-\frac{1}{2}}\bigg(\mathcal{E}((S_{n-1}^{\frac{1}{2}}R S_{n-1}^{\frac{1}{2}})^\frac{1}{2})\bigg)^2 S_{n-1}^{-\frac{1}{2}}$}}
\EndFor
\Ensure{Approximate solution $S_{N}$}
\end{algorithmic}
\end{algorithm}
In Appendix~\ref{initial point} we show that the initial point $S_0=\mathcal{E}(R^\frac{1}{2})^2$ is the solution of~\eqref{first problem} when $R$ commutes with $\mathcal{E}(R^\frac{1}{2})$, i.e., when $[R,\mathcal{E}(R^\frac{1}{2})]=0$. Moreover, we observe numerically that it generally gives a close approximation of the true solution. However, we remark that we prove that the algorithm converges to the solution starting from any initial positive invariant matrix, and we do not actually require that the initial point is $S_0=\mathcal{E}(R^\frac{1}{2})^2$.

The main contribution of our work is to provide a theoretical convergence guarantee of Algorithm~\ref{Fixed-point algo}.
\begin{theorem}
\label{theorem convergence}
Let $R$ be a positive definite matrix. Consider Algorithm~\ref{Fixed-point algo} to solve the
convex optimization problem~\eqref{first problem} and $T$ represent the optimal solution. Then the sequence $\{B(R,S_n)\}$ decreases monotonically, and satisfies
\begin{align}
    B(R, S_n)^2 - B(R, T)^2 
    &\leq \left(1 - \frac{1}{\xi}\right)^n \left(B(R, S_0)^2 - B(R, T)^2\right), \quad \forall n \in \mathbb{N}, 
\end{align}
where 
\begin{align} 
    \xi &= \left( 
        \frac{\lambda_{\max}(\mathcal{E}(R))}{\lambda_{\min}(\mathcal{E}(R))} 
        \frac{\det(\mathcal{E}(R))}{\det(R)} 
    \right)^{3/2}.
\end{align}
  Moreover, when the invariant matrices commute, we have $\xi = (\lambda_{\max}(\mathcal{E}(R))/ \lambda_{\min}(R))^\frac{3}{2}$.
\end{theorem}
We notice that, in general, \( \det(\mathcal{E}(R)) \geq \det(R) \) as a consequence of the concavity of the log-determinant function. This inequality reflects the gap introduced by Jensen's inequality, which can, in unfavorable cases, lead to significantly worse convergence rates compared to the dimension-independent convergence obtained when all matrices commute. Nevertheless, for well-behaved instances, this gap remains modest. Moreover, we note that the stronger result for commuting matrices ultimately stems from the square being an operator monotone function in this case (see Section~\ref{Uniform lower bounds} and~\ref{conclusion} for detailed discussion and further directions).

In Theorem~\ref{theorem convergence}, we assume that the input positive matrix is positive, i.e., it is full-rank. If $R$ is rank-one, it is straightforward to see that the solution of~\eqref{first problem} is given by the eigenvector corresponding to the maximum eigenvalue of  $\mathcal{E}(R)$ multiplied by $\lambda_{\max}(\mathcal{E}(R))$ which yields the value $\Tr(R)-\lambda_{\max}(\mathcal{E}(R))$. For non-full-rank positive semidefinite matrices that are not rank-one, the solution can be approximated arbitrarily well by choosing a close positive definite matrix to which our result applies.  We refer to Appendix~\ref{continuity bound} for more details.

\subsection{Relationship with the Bures-Wasserstein barycenter problem}
\label{remark barycenter}
The Bures-Wasserstein barycenter problem is a special case of the general problem for a specific group representation and a specific input state $R_{AB} \in \mathcal{P}(A\otimes B)$. 
Given a set of positive semidefinite matrices $X_1,...,X_m$ and a weight vector $ \omega = (\omega_1,..., \omega_m)$; i.e., $\omega_j \geq 0$ and $\sum_{j=1}^m \omega_j=1$, the Bures-Wasserstein barycenter problem is (see~\cite{bhatia2019bures})
\begin{equation}
\label{Bures-Wasserstein barycenter}
   \Omega(\omega; X_1,...,X_m) = \argmin_{S \geq 0}\sum\nolimits_{j=1}^m \omega_jB(X_j,S)^2 \,.
\end{equation}
Let us consider an orthonormal basis $\{\phi_k\}_{k=1}^{d_A}$ for the Hilbert space $A$. The Heisenberg-Weyl operators $U_{(l,m)}$ can be defined through the action $U_{(l,m)}\phi_k=e^{\frac{2\pi i k m}{d_A}}\phi_{k+l (\text{mod} \; d_A)}$. Here, $l,m = 1,...,d_A$. Let us consider the group representation $\{U_{(l,m),A} \otimes I_B\}_{l,m=1}^{d_A}$. It is easy to verify that the group averaging on a positive matrix $S_{AB} \in \mathcal{P}(A\otimes B)$ acts as
\begin{equation}
\label{averaging barycenter}
\mathcal{E}(S_{AB}) = \frac{1}{d_A^2}\sum\nolimits_{l,m=1}^{d_A} \big(U_{(l,m),A} \otimes I_B\big) S_{AB} \big(U^\dagger_{(l,m),A} \otimes I_B\big) = \pi_{A} \otimes S_B \,,
\end{equation}
where $\pi_{A} = I_A/d_A$ is the fully mixed state on the space $A$ and $S_B = \Trm_A[S_{AB}]$. Moreover, any invariant positive matrix is of the form $\pi_{A} \otimes S_B$.
Let us set $R_{AB} = \sum_{j=1}^{d_A} d_A \omega_j^2 E^{j,j}_{A} \otimes X_{j,B} \in \mathcal{P}(A \otimes B)$, where $E^{j,j}_{A}$ is the matrix with components $E^{j,j}_{A}(a_1,a_2) = 1$ if $a_1=a_2=j$ and zero otherwise. 
For this specific case, the optimizer of the problem~\eqref{first problem} coincides with the Bures-Wasserstein barycenter~\eqref{Bures-Wasserstein barycenter} with $m=d_A$ and $X_j=X_{j,B}$.
Indeed, we have
\begin{align}
    \argmin_{S_B \geq 0} B\Big( \sum\nolimits_{j=1}^{d_A} d_A\omega_j^2 E^{j,j}_{A} \otimes X_{j,B},\pi_A \otimes S_B\Big)^2 &= \argmin_{S_B \geq 0} \Tr[S_B] - 2 \sum\nolimits_{j=1}^{d_A} \omega_j F(X_{j,B}, S_B)^\frac{1}{2} \\
    &=\argmin_{S_B \geq 0}\sum\nolimits_{j=1}^{d_A} \omega_jB(X_{j,B},S_B)^2 \,,
\end{align}
where we used that $F\big( \sum_{j=1}^{d_A} d_A\omega_j^2 E^{j,j}_{A} \otimes X_{j,B},\pi_A \otimes S_B\big)^\frac{1}{2}= \sum_{j=1}^{d_A} \omega_j F(X_{j,B}, S_B)^\frac{1}{2}$ in the first equality and that constant factors do not affect the solution of the problem in both equalities.
Moreover, a straightforward calculation shows that Algorithm~\ref{Fixed-point algo} recovers the fixed-point iterative algorithm originally introduced in~\cite{alvarez2016fixed}. Explicitly, we obtain
\begin{equation}
S_{AB,n} = \pi_A \otimes S_{B,n} = \pi_A \otimes S_{B,n-1}^{-\frac{1}{2}} \Big(\sum\nolimits_{j=1}^m \omega_j \big(S_{B,n-1}^\frac{1}{2} X_j S_{B,n-1}^\frac{1}{2}\big)^\frac{1}{2} \Big)^2 S_{B,n-1}^{-\frac{1}{2}} \,.
\end{equation}
Moreover, $\mathcal{E}\big(R_{AB}^\frac{1}{2}\big)^2 = \pi_A \otimes \Big(\sum_{j=1}^m \omega_j X_{j,B}^\frac{1}{2}\Big)^2$, where $R_{AB} = \sum_{j=1}^{m} d_A \omega_j^2 E^{j,j}_{A} \otimes X_{j,B}$. Therefore, we can see that the update acts on the $B$ system while the marginal $A$ remains constant. Hence, for this specific case, Algorithm~\ref{Fixed-point algo} is equivalent to 
\begin{algorithm}[H] 
\normalsize
\caption{Fixed-point iterative algorithm for the Bures-Wasserstein barycenter}\label{Fixed-point algo Bures-Wasserstein barycenter}
\label{alg:loop 2}
\begin{algorithmic}[1]
\Require{ $X_1,...,X_m$ positive definite matrices, weight vector $(\omega_1,..., \omega_m)$, number of iterations $N$} 
\State {\textbf{Initialize:} {Initial point $S_0=$ {$\big(\sum_{j=1}^m \omega_j X_j^\frac{1}{2}\big)^2$}}}
\For{$n = 1,...,N$}                    
    \State {$S_{n} = $  {$S_{n-1}^{-\frac{1}{2}} \Big(\sum_{j=1}^m \omega_j \big(S_{n-1}^\frac{1}{2} X_j S_{n-1}^\frac{1}{2}\big)^\frac{1}{2} \Big)^2 S_{n-1}^{-\frac{1}{2}}$}}
\EndFor
\Ensure{Approximate solution $S_{N}$}
\end{algorithmic}
\end{algorithm}
As a corollary of Theorem~\ref{theorem convergence}, we obtain the following result.

\begin{corollary}
\label{corollary Bures-Wasserstein barycenter}
Let $X_1,...,X_m$ be a set of positive definite matrices for any $i=1,..,m$ and $ \omega = (\omega_1,..., \omega_m)$ be a weight vector.
Consider the Algorithm~\ref{Fixed-point algo Bures-Wasserstein barycenter} to solve the
Bures-Wasserstein barycenter problem~\eqref{Bures-Wasserstein barycenter} and $T$ represent the optimal solution.
Then, the sequence $\big\{\sum_{j=1}^n \omega_j B(X_j,S_n)^2\big\}$ decreases monotonically, and satisfies
\begin{equation}
     \sum_{j=1}^m \omega_j \left(B(X_j,S_n)^2-  B(X_j,T)^2 \right)\leq \left(1-\frac{1}{\xi}\right)^{n}\sum_{j=1}^m \omega_j \left(B(X_j,S_0)^2- B(X_j,T)^2\right) \,, \quad \forall n \in \mathbb{N} \,,
\end{equation}
where
\begin{equation}
    \xi = \left(\frac{\lambda_{\max}\Big(\sum\nolimits_{j=1}^m \omega_j X_j\Big)}{\lambda_{\min}\Big(\sum\nolimits_{j=1}^m \omega_j X_j\Big)}\frac{\det\big(\sum_{j=1}^m\omega_j X_j\big)}{\exp\Big(\sum_{j=1}^m\omega_j\log\det(X_j)\big)}\right)^\frac{3}{2} \,.
\end{equation}
\end{corollary}
Our analysis provides an explicit, problem-dependent characterization of the constant $\xi$. Although this constant does not yield a dimension-dependent bound in all cases and performs well for well-behaved instances, its worst-case scaling becomes dimension-dependent.

The proof is analogous to that of Theorem~\ref{theorem convergence}. However, a more precise analysis is needed to carefully handle the weights $\omega_i$. We discuss this in Section~\ref{Section corollary}.

\section{The fixed-point equation}
\label{sec: fixed-point equation}
In this section, we show that any solution of~\eqref{first problem} for full-rank states satisfies a fixed-point equation. In fact, we show in Section~\ref{second} that the solution is unique. The fixed-point equation satisfied by the optimizer motivates the fixed-point Algorithm~\ref{Fixed-point algo}. Indeed, the optimizer is the fixed-point of the algorithm's update.
\begin{lemma}
\label{fixed-point equation}
    Let $R$ be a positive definite matrix. Then $T$ is a solution of the problem~\eqref{first problem} if and only if $T$ is positive definite and
    \begin{equation}
     T = T^{-\frac{1}{2}}\big(\mathcal{E}((T^{\frac{1}{2}}R T^{\frac{1}{2}})^\frac{1}{2})\big)^2 T^{-\frac{1}{2}} \,.
\end{equation} 
\end{lemma}
\begin{proof}
The result follows from the necessary and sufficient conditions satisfied by the optimal solution of the problem~\eqref{first problem}. We follow a similar approach to~\cite[Theorem 4]{rubboli2022new} for $\alpha=z=1/2$.  
We first prove that if $R$ is full-rank, then any optimizer $T$ of~\eqref{first problem} must also be full-rank. 
To do that, we show that if $T$ does not satisfy the latter conditions on the support, we can always find a sufficiently small $x$ such that $B_R(T+xI) < B_R(T)$. This contradicts the fact that $T$ is an optimizer. Here, we implicitly used that $T+xI$ is an invariant positive semidefinite matrix since $\mathcal{E}(I)=I$. We can upper bound the derivative along the direction $I$ as
\begin{align}
\frac{\d}{\d x} B_R(T+xI) &= \Tr[I] - \Trm\big[R^\frac{1}{2}( R^{\frac{1}{2}} (T+xI) R^{\frac{1}{2}})^{-\frac{1}{2}} R^\frac{1}{2}\big] \\
& \leq  \Tr[I]-(\Trm[R^{\frac{1}{2}} T R^{\frac{1}{2}}]+x \Tr[R])^{-\frac{1}{2}}\Trm\big[R^\frac{1}{2}\Pi(R^{\frac{1}{2}} T R^{\frac{1}{2}})R^\frac{1}{2}\big] \\
& \qquad - x^{-\frac{1}{2}} \Trm\big[R^\frac{1}{2}(I-\Pi(R^{\frac{1}{2}} T R^{\frac{1}{2}}))R^\frac{1}{2}\big] \Tr[R]^{-\frac{1}{2}} \,,
\end{align}
where we used that $R^{\frac{1}{2}} (T+xI) R^{\frac{1}{2}} \leq \Pi(R^{\frac{1}{2}} T R^{\frac{1}{2}})(\Trm[R^{\frac{1}{2}} T R^{\frac{1}{2}}]+x \Trm[R]) + x (I-\Pi(R^{\frac{1}{2}} T R^{\frac{1}{2}})) \Trm[R]$ and that the power $-1/2$ is operator antimonotone. Here, we denoted with $\Pi(S)$ the projector onto the support of a positive semidefinite matrix $S$. Note that $\Pi(R)=I$ and everything is well-defined since $T+xI$ is full-rank for any $x>0$. The trace $\Trm\big[R^\frac{1}{2}(I-\Pi(R^{\frac{1}{2}} T R^{\frac{1}{2}}))R^\frac{1}{2}\big]$ is different from zero only if $ \text{supp}(R^\frac{1}{2} T R^\frac{1}{2}) \subset \text{supp}(I)$. In this case, the derivative goes to minus infinity as $-x^{-\frac{1}{2}}$ as $x$ vanishes. This means that in the neighborhood of $x=0$ the function $B_R(T+xI)$ is a strictly decreasing (and continuous) function of $x$ and hence we can always find a sufficiently small $x$ such that $B_R(T+xI) < B_R(T)$. This implies that if $T$ is an optimizer, it must be that $\text{supp}(I) \subseteq \text{supp}(R^\frac{1}{2} T R^\frac{1}{2})$ . By noting that we always have that $\supp(R^\frac{1}{2} T R^\frac{1}{2})\subseteq \supp(I)$, we obtain the condition $\text{supp}(R^\frac{1}{2} T R^\frac{1}{2}) = \text{supp}(I)$. Since $R$ is full-rank, this means that $T$ must be full-rank.

Because of the latter observation, the optimum is achieved in the interior of the set.  Since the optimizer is achieved in the set's interior, the Frech\'et derivative along any invariant Hermitian direction at the optimum is equal to zero. 
To demonstrate this, suppose that the derivative in $T$ along a Hermitian direction $Y$ is negative. Then $T$ is not a minimum as $T+x Z$ for sufficiently small $x$ yields a lower value for the Bures distance. Note that $T+x Z$ is contained in the set $\mathcal{P}_G(A)$ since it is positive semidefinite due to the continuity of the eigenvalues and it is invariant under averaging. Conversely, if the derivative in $Z$ along a Hermitian direction $Z$ is positive, we can consider the direction $-Z$ and repeat the argument above. Therefore, the derivative in the positive definite minimum $T$ along any Hermitian direction $Z$ must be zero.

The Frech\'et derivative along an invariant Hermitian direction $Z$ in the invariant point $S$ is 
\begin{align}
DB_R(S)^2(Z) &= \Tr[Z] - \Tr[R^\frac{1}{2}( R^{\frac{1}{2}} S R^{\frac{1}{2}})^{-\frac{1}{2}} R^\frac{1}{2}Z] \\
&= \Tr[\left(I-R^\frac{1}{2}( R^{\frac{1}{2}} S R^{\frac{1}{2}})^{-\frac{1}{2}} R^\frac{1}{2}\right)Z].
\end{align}
We use that $f(Y^{\dagger} Y)Y^{\dagger} = Y^{\dagger} f(Y Y^{\dagger})$ for any function $f$ and set $Y=S^\frac{1}{2} R^\frac{1}{2}$ to rewrite 
\begin{equation}
    R^\frac{1}{2}( R^{\frac{1}{2}} S R^{\frac{1}{2}})^{-\frac{1}{2}} R^\frac{1}{2} = S^{-\frac{1}{2}}(S^{\frac{1}{2}}RS^{\frac{1}{2}})^\frac{1}{2}S^{-\frac{1}{2}} \,.
\end{equation}
Hence, we obtain
\begin{equation}
DB_R(S)^2(Z) = \Tr[\left(I- S^{-\frac{1}{2}}(S^{\frac{1}{2}}RS^{\frac{1}{2}})^\frac{1}{2}S^{-\frac{1}{2}} \right) Z].
\end{equation} 
 Since $DB_R(S)^2(Z) = \langle \nabla B_R(S)^2,Z \rangle$, we have that $\nabla B_R(S)^2 =I- S^{-\frac{1}{2}}(S^{\frac{1}{2}}RS^{\frac{1}{2}})^\frac{1}{2}S^{-\frac{1}{2}}$. To derive the fixed point equation, based on the previous discussion, we set the derivative at the optimal point $T$ equal to zero. Since $Z = \mathcal{E}(X)$ for some Hermitian matrix $X$, using the cyclicity of the trace, we can rewrite the condition as 
\begin{equation}
\Tr[\left(I- \mathcal{E}(T^{-\frac{1}{2}}(T^{\frac{1}{2}}R T^{\frac{1}{2}})^\frac{1}{2}T^{-\frac{1}{2}}) \right) X] = 0 \,,
\end{equation}
for all Hermitian $X$. Since this holds for all Hermitian $X$, and the trace of a matrix with all Hermitian matrices being zero implies the matrix is zero, the gradient must vanish. Hence, we obtain the fixed-point relation
\begin{equation}
\mathcal{E}(T^{-\frac{1}{2}}(T^{\frac{1}{2}}R T^{\frac{1}{2}})^\frac{1}{2}T^{-\frac{1}{2}}) = T^{-\frac{1}{2}}\mathcal{E}((T^{\frac{1}{2}}R T^{\frac{1}{2}})^\frac{1}{2})T^{-\frac{1}{2}} = I \,,
\end{equation}
where we used that an invariant positive semidefinite matrix commutes with each of the unitaries of the averaging. 
We then multiply both terms by $T^\frac{1}{2}$ on both sides, square the expression, and multiply both sides by $T^{-\frac{1}{2}}$. This gives us the fixed-point equation for the optimum $T$
\begin{equation}
T = T^{-\frac{1}{2}}\bigg(\mathcal{E}((T^{\frac{1}{2}}R T^{\frac{1}{2}})^\frac{1}{2})\bigg)^2 T^{-\frac{1}{2}} \,.
\end{equation} 
\end{proof}
We now present the proof of convergence. In particular, we show that starting from any invariant matrix, the fixed-point iteration Algorithm~\ref{Fixed-point algo} converges to the fixed point.

\section{Proof of Theorem 1}
\label{proof of Theorem 1}
In this section, we provide a proof of Theorem~\ref{theorem convergence}. 
The proof relies on two inequalities, which we derive independently. The first is based on a matrix version of Hölder's inequality and establishes a quantitative contraction property of the Bures distance between the input matrix and each iterate under the fixed-point iteration. This provides an alternative approach to that used in the specific case of the Bures-Wasserstein barycenter in~\cite{alvarez2016fixed} and~\cite{bhatia2019bures}, which combines optimal transport inequalities with matrix inequalities. We remark that this inequality alone is sufficient to show asymptotic convergence as shown in~\cite{alvarez2016fixed} and~\cite{bhatia2019bures}. The second one is a Polyak-\L ojasiewicz-type inequality which holds for strongly convex functions. By combining these two inequalities, we establish Theorem~\ref{theorem convergence}, providing a quantitative bound on the convergence rate as a function of the number of iterations. Finally, we note that the result we prove is actually slightly stronger: the algorithm converges to the optimum from any initial positive invariant matrix.

\subsection{First inequality}
The first inequality follows directly from Hölder's inequality. This result demonstrates that the Bures distance to the input matrix monotonically decreases under each fixed-point iteration.

We have that for two Hermitian matrices $A$ and $C$ it holds
\begin{equation}
\label{Watrous inequality}
F(A A^\dagger, C C^\dagger)^\frac{1}{2} = \|A^\dagger C\|_1 \geq |\Tr\, [A^\dagger C]| \,.
\end{equation}
The equality is proved in~\cite[Lemma 3.21]{watrous2018theory}. The inequality is a consequence of Holder's inequality $|\langle N,M\rangle| \leq \|M\|_1\|N\|_{\infty}$  with $M=A^\dagger C$ and the identity matrix  $N = I$ (see e.g.~\cite[Equation 1.173]{watrous2018theory}). In the following, for compactness, we denote the action of the algorithm's update as $K(S) :=S^{-\frac{1}{2}}\big(\mathcal{E}((S^{\frac{1}{2}}R S^{\frac{1}{2}})^\frac{1}{2})\big)^2 S^{-\frac{1}{2}}$. We then set
\begin{equation} 
A= S^{-\frac{1}{2}}(S^{\frac{1}{2}}RS^{\frac{1}{2}})^\frac{1}{2}, \quad C=S^{-\frac{1}{2}}\mathcal{E}((S^{\frac{1}{2}}RS^{\frac{1}{2}})^\frac{1}{2}) \,.
\end{equation}
so that $AA^\dagger = R$, and $CC^\dagger = K(S)$.
The inequality~\eqref{Watrous inequality} gives
\begin{align}
&F(R,K(S))^\frac{1}{2} \geq  \Big|\Trm \big[(S^{\frac{1}{2}}RS^{\frac{1}{2}})^\frac{1}{2}S^{-1} \mathcal{E}((S^{\frac{1}{2}}RS^{\frac{1}{2}})^\frac{1}{2})\big] \Big|= \Tr[K(S)] \,.
\end{align}
In the equality, we used that $S^{-1}$  is invariant and hence it commutes with any of the unitaries of the group averaging. Therefore, an additional averaging could be attached to $(S^{\frac{1}{2}}RS^{\frac{1}{2}})^\frac{1}{2}$. Finally, we used the cyclicity of the trace and the fact that $K(S)$ is positive, and hence we can drop the absolute value. Moreover, the latter inequality implies that
\begin{align}
&F(R,K(S))^\frac{1}{2} \geq \Tr[K(S)] \\
\Longleftrightarrow \quad & \Tr[S] + \Tr[K(S)] - 2F(S,K(S))^\frac{1}{2} \\
& \quad \quad \leq \Tr[R] + \Tr[S] - 2F(R,S)^\frac{1}{2} 
- \big(\Tr[R] + \Tr[K(S)] - 2F(R,K(S))^\frac{1}{2}\big)\\
\label{fist inequality}
\Longleftrightarrow \quad & B(S, K(S))^2 \leq B_R( S)^2 - B_R(K(S))^2 \,,
\end{align}
where the second equivalence, we used that $F(S,K(S))=F(R,S)$. Indeed, using the definition of fidelity, we have
\begin{align}
F(S,K(S))^\frac{1}{2} &= \Tr\big[(S^\frac{1}{2}K(S)S^\frac{1}{2})^\frac{1}{2}\big] 
 =\Tr\big[\mathcal{E}((S^{\frac{1}{2}}R S^{\frac{1}{2}})^\frac{1}{2})\big]  = \Tr\big[(S^{\frac{1}{2}}R S^{\frac{1}{2}})^\frac{1}{2}\big]  = F(R,S)^\frac{1}{2}\,,
\end{align} 
where we used that the group averaging does not change the trace as can be easily proven using the cyclicity of the trace and that $U_gU_g^\dagger = I$ for any $g \in G$.
Since the Bures distance on the left-hand side is always positive, inequality~\eqref{fist inequality} also shows that the sequence produced by the fixed-point iteration has a monotonically decreasing Bures distance to $R$.

 In the following, we show that by combining inequality~\eqref{fist inequality} with strong convexity arguments, it is possible to derive a quantitative statement about the convergence in the number of iterations. 

\subsection{Uniform lower bounds for strong convexity}
 \label{Uniform lower bounds}
 Before proving the second key inequality in our argument, we first discuss the strong convexity of the Bures distance, which will play a crucial role in our proof.

The Bures distance is strongly convex in the space of positive definite matrices (see the discussion in~\cite[Section 1]{bhatia2018strong}). Indeed, we have
\begin{theorem}[{\cite[Theorem 1]{bhatia2018strong}}]
\label{strong convexity}
Let $R,S,Y$ be positive matrices such that $R,S,Y \in [\alpha I, \beta I]$. Then, we have
\begin{equation}
\label{bound Hessian fidelity}
    B_R(Y)^2 \geq B_R(S)^2 + \langle\nabla B_R(S)^2,(Y-S)\rangle +\frac{\mu} {2}\|Y-S\|^2 \,.
\end{equation}
Here, $\mu := \alpha^\frac{1}{2}/(2\beta^\frac{3}{2})$. Hence, the Bures distance is $\mu$-strongly convex.
\end{theorem}

We now show that at any iterations it holds $S_n = K(S_{n-1}) \in [\alpha I, \beta I]$ for any value of $n$. Here, $\alpha >0$ and $\beta< \infty$ are two constants that depend only on the input of the problem $R$.
Furthermore, since \( T = K(T) \), it follows that \( T \in [\alpha I, \beta I] \). Therefore, in our subsequent analysis, we can apply the strong convexity result stated in Theorem~\ref{strong convexity}.

We present two proofs: one for the general case and another for the special case in which all invariant matrices commute. In the latter setting, we are able to obtain stronger bounds. The key reason we achieve tighter estimates when considering commuting invariant matrices in Subsection~\ref{stronger bounds} is that the squaring function is operator monotone only for commuting matrices.

Finally, we observe that, since any optimal solution can be restricted to the compact interval above, and Theorem~\ref{strong convexity} ensures that the Bures distance is strongly convex over this interval, the solution $T$ is unique.

In the remainder of this section, we focus on the general case. To address the non-monotonicity of the square function, we instead make use of the determinant. However, this approach yields a problem-dependent bound that performs reasonably well for well-behaved problems, but in the worst case leads to a dimension-dependent convergence rate.
 
\begin{lemma}
\label{lemma: non commuting}
Let $R$ be a positive definite matrix. Then, for any positive definite matrix $S$, it holds that
\begin{equation}
    K(S) \in \left[\lambda_{\min}(\mathcal{E}(R))\frac{\det(R)}{\det(\mathcal{E}(R))} I,\lambda_{\max}(\mathcal{E}(R))I\right] \,,
\end{equation}
where $K(S)= S^{-\frac{1}{2}}\big(\mathcal{E}((S^{\frac{1}{2}}R S^{\frac{1}{2}})^\frac{1}{2})\big)^2 S^{-\frac{1}{2}}$ is the iteration of Algorithm~\ref{alg:loop}.
\end{lemma}
\begin{proof}
The upper bound follows from Jensen's inequality for the operator convex function $t \mapsto t^2$ for unital maps~\cite{hansen2003jensen}.
Indeed, the group averaging is an unital map. We have
\begin{align}
&\big(\mathcal{E}((S^{\frac{1}{2}}R S^{\frac{1}{2}})^\frac{1}{2})\big)^2  \leq \mathcal{E}((S^{\frac{1}{2}}R S^{\frac{1}{2}})) = S^\frac{1}{2}\mathcal{E}(R)S^\frac{1}{2} \,,
\end{align}
where in the last equality we used that since $S$ is invariant, it holds that $[U_g,S]=0$ for each element of the group averaging map.
Hence, we obtain
\begin{equation}
\label{upper bound}
    K(S) = S^{-\frac{1}{2}} \big(\mathcal{E}((S^{\frac{1}{2}}R S^{\frac{1}{2}})^\frac{1}{2})\big)^2 S^{-\frac{1}{2}}  \leq \mathcal{E}(R) \,.
\end{equation}
Since the above holds as an operator inequality, it follows from Corollary \cite[III.1.2]{bhatia1997matrix} that the eigenvalues are also ordered, i.e.,
\begin{equation}
\label{ordering eigenvalues}
    \lambda_{k}(K(S)) \leq \lambda_k(\mathcal{E}(R)) \quad \forall k = 1, \ldots, d.
\end{equation}
Hence, we obtain that $\lambda_{\max}(K(S)) \leq \lambda_{\max}(\mathcal{E}(R)) \leq \lambda_{\max}(R)$ since $R \leq \lambda_{\max}(R) I$ and the averaging preserves the operator inequality, i.e., if $A \geq B$ then it also holds that $\mathcal{E}(A) \geq \mathcal{E}(B)$. 

Let us now derive a lower bound for the minimum eigenvalue. 
To lower bound the minimum eigenvalue, we apply the concavity of the \(\log\det\) function in conjunction with Jensen's inequality and the multiplicativity property of the determinant, \(\det(AB) = \det(A)\det(B)\), to obtain
\begin{align}
	\log\det(K(S)) &= 2 \log\det\big(\mathcal{E}((S^\frac{1}{2}RS^\frac{1}{2})^\frac{1}{2})\big)-\log\det(S)\\
	& \geq 2 \log\det\big((S^\frac{1}{2}RS^\frac{1}{2})^\frac{1}{2}\big)-\log\det(S) \\
	& =\log\det(R) \,.
\end{align}
Equivalently, for the minimum eigenvalue, if the eigenvalues are sorted in increasing order, it implies
\begin{equation}
\label{Intermediate}
\log{\lambda_{\min}(K(S))} \geq \log{\lambda_{\min}(R)}+ \sum_{i=2}^d \big(\log{\lambda_i(R)}-\log{\lambda_i(K(S))}\big) \,.
\end{equation}
Using once again inequality~\eqref{ordering eigenvalues}, and taking the exponential on both sides, we finally obtain
\begin{equation}
    \lambda_{\min}(K(S)) \geq \lambda_{\min}(R)\prod_{i=2}^d \frac{\lambda_i(R)}{\lambda_i(\mathcal{E}(R))} = \lambda_{\min}(\mathcal{E}(R))\frac{\det(R)}{\det(\mathcal{E}(R))} \,.
\end{equation}
\end{proof}
We then set $\alpha = \lambda_{\min}(\mathcal{E}(R))\det(R)/\det(\mathcal{E}(R))$ and $\beta = \lambda_{\max}(\mathcal{E}(R))$. 
This shows that there exists some constant $\alpha>0$ and $\beta < \infty$ such that $S_n \in [\alpha I, \beta I]$ for any $n \geq 1$ and therefore the problem does not get arbitrarily ill-conditioned. This also shows that the inverses are always well-defined at each step of the iteration since all the matrices are full rank.

 \subsubsection{Stronger bounds for commuting invariant matrices}
 \label{stronger bounds}
In this section, we show that when the invariant matrices commute, as in the case of a projective unitary representation in which all irreducible components have multiplicity one, the lower bound can be strengthened, and we can establish that \( S_n \in [\lambda_{\min}(R) I, \lambda_{\max}(\mathcal{E}(R)) I] \). This includes, for example, the setting of the fidelity of coherence (see Section~\ref{quantum info}).
\begin{lemma}
\label{uniform commuting}
Let $R$ be a positive definite matrix. Then, for any positive definite matrix $S$ such that $\mathcal{E}\big((S^\frac{1}{2}RS^\frac{1}{2})^\frac{1}{2}\big)$ and $S$ commute, it holds that $K(S) \in [ \lambda_{\min}(R)I,\lambda_{\max}(\mathcal{E}(R))I]$. Here, $K(S)= S^{-\frac{1}{2}}\big(\mathcal{E}((S^{\frac{1}{2}}R S^{\frac{1}{2}})^\frac{1}{2})\big)^2 S^{-\frac{1}{2}}$ is the iteration of Algorithm~\ref{alg:loop}.
\end{lemma}
\begin{proof}
For the upper bound, the proof proceeds in the same way as the proof of Lemma~\ref{lemma: non commuting}. However, to derive the lower bound for the minimum eigenvalue, instead of using the determinant function, we note that for commuting matrices the square is operator monotone. 

We have that
\begin{align}
& R \geq \lambda_{\min}(R) I \\
\implies \quad & S^\frac{1}{2} R S^\frac{1}{2} \geq  \lambda_{\min}(R) S \\
\label{srom}
\implies \quad & (S^\frac{1}{2} R S^\frac{1}{2})^\frac{1}{2} \geq  \lambda_{\min}(R)^\frac{1}{2} S^\frac{1}{2} \\
\label{averaging om}
\implies \quad & \mathcal{E}((S^\frac{1}{2} R S^\frac{1}{2})^\frac{1}{2}) \geq  \lambda_{\min}(R)^\frac{1}{2} S^\frac{1}{2} \, , 
\end{align}
where in~\eqref{srom} we used that the square root is operator monotone. Moreover, in~\eqref{averaging om} we used that $S$ is invariant under averaging. Since we assume that the two terms in the inequality commute, we can apply the monotonicity of the square to obtain
\begin{equation}
(\mathcal{E}((S^\frac{1}{2} R S^\frac{1}{2})^\frac{1}{2}))^2 \geq  \lambda_{\min}(R) S \,.
\end{equation}
This implies that
\begin{equation}
 K(S) = S^{-\frac{1}{2}} \Big(\mathcal{E}((S^{\frac{1}{2}}R S^{\frac{1}{2}})^\frac{1}{2})\Big)^2 S^{-\frac{1}{2}}  \geq \lambda_{\min}(R) I \,.
\end{equation}
\end{proof}
Hence, in this specific case, we can set $\alpha = \lambda_{\min}(R)$ and $\beta = \lambda_{\max}(\mathcal{E}(R))$. 

\subsection{Second inequality}
\label{second}
The PL inequality serves as the second key ingredient in proving our main result. For strongly convex functions, this inequality bounds the difference between a function's value at any point and its optimal value in terms of the gradient norm at that point. We follow a similar approach to~\cite[Appendix B]{karimi2016linear}. However, in our case, the optimization is performed over all invariant matrices, and it is not an unconstrained problem over $\mathbb{R}^d$. Nevertheless, we show below that a similar inequality can also be derived in this context.

From Theorem~\ref{strong convexity}, we know that for positive definite matrices $R,S,Y$ in the compact interval $ [\alpha I,$ $ \beta I]$, the Bures distance is strongly convex 
\begin{equation}
\label{intermediate}
     B_R(Y)^2 \geq B_R(S)^2 + \Tr[\nabla B_R(S)^2(Y-S)] +\frac{\mu} {2}\Tr [(Y-S)^\dagger (Y-S)] :=G(Y) \,.
\end{equation}
with $\mu = \alpha^\frac{1}{2}/(2\beta^\frac{3}{2})$.
We then minimize both sides of the equation with respect to $Y$ in the set of positive invariant semidefinite matrices. Note that the minimization does not change the direction of the inequality. We lower bound the minimization on the r.h.s. with a minimization over the invariant Hermitian matrices $Y$.  The derivative of the r.h.s. of~\eqref{intermediate} along a Hermitian invariant $Z$ gives
\begin{equation}
    DG(Y)(Z) = \Tr[\nabla B_R(S)^2Z] +\frac{\mu}{2}(2\Tr[YZ]-2\Tr[SZ]) \,.
\end{equation}
We have that $D^2G(Y)(Z,Z)= \mu \Tr[Z Z] > 0$ if $\mu > 0$. Therefore, the function $G(Y)$ will have only one minimum. To find the minimum, since the set of Hermitian matrices is open, we can set the derivative equal to zero. It is easy to see that the minimum in the set of Hermitian invariant matrices is achieved by $Y^*=S-\frac{1}{\mu}\mathcal{E}(\nabla B_R(S)^2)$. The averaging is to ensure that $Y^*$ is invariant. Upon substitution, the r.h.s. becomes
\begin{equation}
    G(Y^*) = B_R(S)^2 - \frac{1}{2 \mu}\| \mathcal{E}(\nabla B_R(S)^2)\|^2 \,.
\end{equation}
The l.h.s just yields $B_R(T)^2$. We then obtain 
\begin{equation}
\label{second inequality}
    B_R(S)^2 - B_R(T)^2 \leq \frac{1}{2 \mu}\|\mathcal{E}(\nabla B_R(S)^2)\|^2 \,,
\end{equation}
where we denoted with $T$ is the optimal point and $\mathcal{E}(\nabla B_R(S)^2) = I - S^{-\frac{1}{2}}\mathcal{E}((S^{\frac{1}{2}} R S^{\frac{1}{2}})^\frac{1}{2}) S^{-\frac{1}{2}}$. 

\subsection{Combining the inequalities}
\label{combining}
Here, we show that the inequalities~\eqref{fist inequality} and~\eqref{second inequality} imply that the algorithm converges exponentially fast in the number of iterations. We first note that $S \geq \alpha I$ implies that
\begin{align}
\label{middle inequality}
\| \mathcal{E}(\nabla{B_R( S)^2}) \|^2 \leq \alpha^{-1} \Tr[ \mathcal{E}(\nabla{B_R(S)^2}) S \mathcal{E}(\nabla{B_R(S)^2})] = \alpha^{-1} B( S,K(S))^2\,.
\end{align}

We then obtain
\begin{align}
 B_R( S)^2-B_R( T)^2 & \leq \frac{1}{2\mu}\| \mathcal{E}(\nabla{B_R( S)^2}) \|^2 \\
& \leq \frac{1}{2\mu \alpha} B( S, K(S))^2 \\
& \leq \frac{1}{2\mu \alpha} (B_R( S)^2 - B_R( K(S))^2) \,.
\end{align}
The first inequality follows from~\eqref{second inequality}, the second inequality from~\eqref{middle inequality}, and in the last inequality we used~\eqref{fist inequality}. If we denote $a_n =  B_R(S_n)^2-B_R(T)^2$ and $\xi=(2\mu \alpha)^{-1}$ we can rewrite the above inequality as  $a_n \leq  \xi (a_n-a_{n+1})$. The latter recursive relation gives 
\begin{equation}
\label{convergence}
 B_R(S_n)^2-B_R(T)^2 \leq \left(1-\frac{1}{\xi}\right)^{n}(B_R( S_0)^2-B_R(T)^2) \,,
\end{equation}
which proves an exponentially fast convergence of the value of the Bures distance in the number of iterations. From Theorem~\ref{strong convexity}, we have that $\mu = \alpha^\frac{1}{2}/(2\beta^\frac{3}{2})$. Moreover, from Lemma~\ref{lemma: non commuting}, in the general case, and Lemma~\ref{uniform commuting} when all the invariant matrices commute, we obtain the value of $\xi$ in Theorem~\ref{theorem convergence}.

\begin{remark}
The constants appearing on the r.h.s. of equation~\eqref{convergence} result from a worst-case argument. We observe that the actual constants are much
smaller in practice. Whether or not a tighter analysis, at least in some regimes, could lead to better constants is still an important open problem. In particular, we note numerically that the convergence is still very fast even for very small minimum eigenvalues.
\end{remark}

\begin{remark}
The initial point $S_0 = \mathcal{E}(R^\frac{1}{2})^2$ satisfies $S_0 \in [\lambda_{\min}(R) I, \lambda_{\max}(\mathcal{E}(R)) I]$. We numerically observe that $S_n \in [\lambda_{\min}(R) I, \lambda_{\max}(\mathcal{E}(R)) I]$ for any $n$ which suggests that 
the lower bound in Lemma~\ref{lemma: non commuting} is too conservative. In practice, we can overcome this problem by evaluating the eigenvalues at each iteration and verifying that $S_n \in [\lambda_{\min}(R) I, \lambda_{\max}(\mathcal{E}(R)) I]$. This allows us to use the PL inequality that follows from the latter stronger condition to certify closeness to the solution. In Section~\ref{numerical comparison}, we follow this method.
\end{remark}

\bigskip
Since the Bures distance is strongly convex (see Theorem~\ref {strong convexity}), it follows that the sequence $\{S_n\}$ converges to the optimal solution $T$. 
 Indeed, if we note that the gradient is zero at the minimum, $\mu$-strong convexity of the Bures distance together with the result just obtained gives 
\begin{equation}
    \|S_n-T\|^2 \leq \frac{2}{\mu} \left(B(R,S_n)^2-B(R,T)^2 \right) \leq \frac{2}{\mu} \left(1-\frac{1}{\xi}\right)^{n}(B(R, S_0)^2-B(R,T)^2) \,.
\end{equation}

\section{Proof of Corollary ~\ref{corollary Bures-Wasserstein barycenter}}
\label{Section corollary}
In this section, we prove Corollary~\ref{corollary Bures-Wasserstein barycenter}, which provides the specific convergence result for the Bures-Wasserstein barycenter case.
The proof closely follows the structure of Theorem~\ref{theorem convergence}, with the only difference being that Lemma~\ref{lemma: non commuting} is replaced by the following result. In particular, additional care is required to properly handle the weights.
\begin{lemma}
Let $X_1,...,X_m$ be a set of positive definite matrices for any $i=1,..,m$ and $ \omega = (\omega_1,..., \omega_m)$ be a weight vector.
Then, for any positive definite matrix $S$, it holds that 
\begin{equation}
    G(S) \in \left[\lambda_{\min}\Big(\sum\nolimits_{j=1}^m \omega_j X_j\Big)\frac{\exp\Big(\sum_{j=1}^m\omega_j\log\det(X_j)\big)}{\det\big(\sum_{j=1}^m\omega_j X_j\big)}I,\lambda_{\max}\big(\sum\nolimits_{j=1}^{m}\omega_j X_j\big)I\right] \,,
\end{equation}
 where $G(S) = S^{-\frac{1}{2}} \big(\sum_{j=1}^m \omega_j \big(S^\frac{1}{2} X_j S^\frac{1}{2}\big)^\frac{1}{2} \big)^2 S^{-\frac{1}{2}}$ is the iteration of Algorithm~\ref{alg:loop 2}.
\end{lemma}
\begin{proof}
The proof is completely analogous to the one of Lemma~\ref{lemma: non commuting}. However, to avoid the dependence on the probabilities $\omega_j$, a slightly different approach is needed. 

The upper bound follows from the operator convexity of the square function. We have
\begin{align}
\label{barycenter upper bound}
&G(S) = S^{-\frac{1}{2}} \left(\sum\nolimits_{j=1}^m \omega_j \Big(S^\frac{1}{2} X_j S^\frac{1}{2}\Big)^\frac{1}{2} \right)^2 S^{-\frac{1}{2}} \leq S^{-\frac{1}{2}} \left(\sum\nolimits_{j=1}^m \omega_j S^\frac{1}{2} X_j S^\frac{1}{2}\right) S^{-\frac{1}{2}} = \sum\nolimits_{j=1}^{m}\omega_j X_j\,.
\end{align}
Hence, as for the general case, we obtain that for each eigenvalue
\begin{equation}
\label{ordering eigenvalues 2}
    \lambda_{k}(G(S)) \leq \lambda_k\Big(\sum\nolimits_{j=1}^{m}\omega_j X_j\Big) \quad \forall k = 1, \ldots, d.
\end{equation}
which gives the desired bound for the maximum eigenvalue. 

To bound the minimum eigenvalue, we apply Jensen's inequality and the multiplicativity property of the determinant, \(\det(AB) = \det(A)\det(B)\), to obtain
\begin{align}
	\log\det(G(S)) &= 2 \log\det\Big(\sum\nolimits_{j=1}^m \omega_j \Big(S^\frac{1}{2} X_j S^\frac{1}{2}\Big)^\frac{1}{2}\Big)-\log\det(S)\\
	& \geq 2 \sum\nolimits_{j=1}^m\omega_j\log\det\big((S^\frac{1}{2}X_jS^\frac{1}{2})^\frac{1}{2}\big)-\log\det(S) \\
	& =\sum\nolimits_{j=1}^m\omega_j\log\det(X_j) \,.
\end{align}
Let \( \lambda_i(C) \) denote the eigenvalues of a positive matrix \( C \), ordered increasingly with respect to the index \( i \).
From the above, we obtain for the minimum eigenvalue
\begin{align}
\log{\lambda_{\min}(G(S))}& \geq \sum\nolimits_{j=1}^m\omega_j\log\lambda_{\min}(X_j) + \sum_{i=2}^d \bigg(\sum\nolimits_{j=1}^m\omega_j\log\lambda_{i}(X_j)-\log{\lambda_i(G(S))}\bigg) \,.
\end{align}
Using inequality~\eqref{ordering eigenvalues 2}, and taking the exponential on both sides, we finally obtain
\begin{equation}
    \lambda_{\min}(G(S)) \geq \prod_{j=1}^m\lambda_{\min}(X_j)^{\omega_j}\prod_{i=2}^d \frac{\prod_{j=1}^m\lambda_i(X_j)^{\omega_j}}{\lambda_i(\sum_{j=1}^m \omega_j X_j)} = \lambda_{\min}\Big(\sum\nolimits_{j=1}^m \omega_j X_j\Big)\frac{\exp\Big(\sum_{j=1}^m\omega_j\log\det(X_j)\big)}{\det\big(\sum_{j=1}^m\omega_j X_j\big)} \,.
\end{equation}

\end{proof}

\section{Application to quantum information theory}
\label{quantum info}
In quantum information theory, we are often concerned with the maximization of the fidelity function over some constrained convex set. Since the fidelity function is connected to the Bures distance, 
our algorithm finds several applications in quantum information theory. Explicitly, in the latter setting, we are often interested in the problem
\begin{align}
\label{second problem}
\argmax_{\sigma \in \mathcal{S}_G(A)}F(R,\sigma) \,.
\end{align}
The solutions of the problem~\eqref{first problem} and~\eqref{second problem} are the same up to a normalization
factor. To see this, let us set $S = k \sigma$ where $k\geq 0$ is a positive constant and $\sigma$ is a quantum state, i.e., $\Tr[\sigma]=1$. The Bures distance between $R$ and the set of positive invariant matrices can be written as
\begin{align}
    \min_{S \in \mathcal{P}_G(A)} B(R, S)^2  &= \min_{\sigma \in \mathcal{S}_G(A), \, k \geq 0} B(R, k\sigma)^2 \\
    &= \min_{\sigma \in \mathcal{S}_G(A), \, k \geq 0} \Tr(R) + k -2k^\frac{1}{2}F(R,\sigma)^\frac{1}{2} \\
    &  = \min_{\sigma \in \mathcal{S}_G(A)} \Tr(R) -F(R,\sigma) \,,
\end{align}
where for the second equality we used that, for any $\sigma$, the optimal $k$ is given by $k = F(R, \sigma)$ from standard calculus. 
Therefore, the solution $T$ of the Bures minimization~\eqref{first problem} and the solution $\tau$ of the fidelity maximization~\eqref{second problem}, are
connected through the relationship $T=k^*\tau$ where $k^* = F(R,\tau)$.

We now list some applications in quantum information theory. We use the standard bra-ket notation to denote vectors and matrices as commonly used in quantum information theory. Moreover, for other basic concepts such as the purification of a quantum state and quantum channels, we refer to the standard textbook~\cite{nielsen2001quantum}.
\begin{itemize}
\item \emph{Fidelity of asymmetry}. In resource theories, the above quantity is known as the fidelity of asymmetry. The fidelity of asymmetry quantifies the asymmetry resource of a quantum state with respect to the fidelity distance between the state and the set of invariant states. In~\cite{laborde2021testing}, the authors propose a quantum algorithm to test the symmetry of a quantum state whose acceptance probability is given by the fidelity of asymmetry.  This endows the latter quantity with an operational meaning. The fidelity of asymmetry is
\begin{align}
\max_{\sigma \in \mathcal{S}_G(A)}F(\rho,\sigma) \,.
\end{align}

A particular instance is given by the resource theory of coherence~\cite{winter2016operational} which plays a prominent role in quantum resource theories.  The fidelity of coherence is an important resource measure in this resource theory~\cite{shao2015fidelity,liu2017new}. Explicitly, let us choose $G$ to be the
cyclic group over $d$ elements with unitary representation
$\{Z(z)\}_{z=0}^{d-1}$, where $Z(z)$ is the generalized Pauli phase-shift unitary defined as $Z(z) := \sum_{j=0}^{d-1} e^{\frac{2\pi i jz}{d}} \ketbra{j}{j}$. In this case, we obtain the fidelity of coherence. 
\item \emph{Max-conditional entropy}. The max-conditional entropy is a common occurrence in quantum Shannon theory~\cite{konig2009operational,tomamichel2015quantum}.
If we choose, for example, the Heisenberg-Weyl operators defined in Remark~\ref{remark barycenter} of Section~\ref{the problem} on the first party of a bipartite system, the fidelity of asymmetry becomes the max-conditional entropy~\cite{konig2009operational,tomamichel2015quantum}
\begin{equation}
H_{\max}(A|B)_{\rho}:= \max_{\sigma_{B} \in \mathcal{S}(B)} \log{F(\rho_{AB}, I_A \otimes \sigma_{B})} \,.
\end{equation}
Indeed, the group averaging acts as $\mathcal{E}(S_{AB}) = \pi_{A} \otimes S_B$.
Moreover, the max-conditional entropy is closely related to the order-$1/2$ sandwiched
R\'enyi mutual information~ \cite{beigi2013sandwiched,gupta2015multiplicativity}
\begin{equation}
I_{\max}(A:B)_\rho:= \max_{\sigma_{B} \in \mathcal{S}(B)} \log{F(\rho_{AB}, \rho_A \otimes \sigma_{B})} \,.
\end{equation} 
Indeed, we can write 
\begin{align}
    F(\rho_{AB}, \rho_A \otimes \sigma_{B}) &= \Tr(\Big((\rho_A \otimes \sigma_{B})^\frac{1}{2} \rho_{AB} (\rho_A \otimes \sigma_{B})^\frac{1}{2}\Big)^\frac{1}{2}) \\
    &= \Tr(\Big((I_A \otimes \sigma_{B})^\frac{1}{2} (\rho_A \otimes I_B)^\frac{1}{2} \rho_{AB} (\rho_A \otimes I_B)^\frac{1}{2} (I_A \otimes \sigma_{B})^\frac{1}{2}\Big)^\frac{1}{2}) \\
    &=F((\rho_A \otimes I_B)^\frac{1}{2} \rho_{AB} (\rho_A \otimes I_B)^\frac{1}{2}, I_A \otimes \sigma_B) \,.
\end{align}
Therefore the solution of the max-conditional entropy for the state $(\rho_A \otimes I_B)^\frac{1}{2} \rho_{AB} (\rho_A \otimes I_B)^\frac{1}{2}$ solves the order-$1/2$ sandwiched
R\'enyi mutual information of the state $\rho_{AB}$.

As discussed at the end of Section~\ref{the problem}, the Bures-Wasserstein barycenter problem corresponds to the max-conditional entropy one for classical-quantum states. Hence, the latter could be seen as a `fully quantum' version of the former.
\item \emph{Geometric measure of entanglement of maximally correlated states}.
The maximally correlated states generalize the notation of pure bipartite states and are often considered in the literature due to their properties. For example, the distillable entanglement has a closed-form expression~\cite{horodecki2000limits}, and several resource monotones become additive if at least one of the two states is of this form~\cite{rubboli2022new}. A \textit{maximally correlated state} is a bipartite state of the form~\cite{rains1999bound}
\begin{equation}
\label{form}
\rho_{AB} = \sum_{jk}\rho_{jk}|j,j\rangle \! \langle k,k|_{AB} \,.
\end{equation}
The geometric measure of entanglement is $E_G(\rho) := 1-F_s(\rho)$, where $F_s(\rho)=\max_{\sigma \in \text{SEP}}F(\rho,\sigma)$ is the fidelity of separability~\cite{wei2003geometric} (see also~\cite{streltsov2010linking} for its connection with the convex-roof formulation). Here, we denote with $\text{SEP}$ the set of separable states.

The max-conditional entropy and the geometric measure of entanglement of maximally correlated states are connected through the relation $E_G(\rho_{AB}) = 1-\max_{\sigma_B \in \mathcal{S}(B)} F(\rho_{AB},I_A \otimes \sigma_B)$~\cite[Theorem 1]{zhu2017coherence}. The problem is then equivalent to the computation of the max-conditional entropy discussed above.
\item \emph{Quantum mean state problem}. 
We first define the quantum mean state problem introduced in~\cite{afham2022quantum}.
Given a collection
of quantum states $\{\rho_1, . . . , \rho_n\} \subset \mathcal{S}(A)$ and a probability vector $p=(p_1,...,p_n)$, the quantum mean state problem is
\begin{equation}
    \argmax_{\sigma \in \mathcal{S}(A)} \sum_{i=1}^np_iF(\rho_i,\sigma)^\frac{1}{2}\,.
\end{equation}
The quantum mean state is therefore the closest state (with respect to the square-root fidelity) to an ensemble of quantum states. In~\cite{afham2022quantum}, the authors discuss some applications in Bayesian quantum tomography.
Moreover, in~\cite[Section 4.2]{afham2022quantum} they show that the Bures-Wasserstein barycenter problem is equivalent to the quantum mean state problem. 
\item \emph{Quantum error precompensation for group averaging channels}. We first define the quantum error precompensation problem introduced in~\cite{zhang2022quantum}.  Given a quantum channel $\mathcal{E}$ and a target state $\rho_t$, the quantum error precompensation problem is
\begin{equation}
    \argmax_{\rho_{\rm{in}} \in \mathcal{S}(A)} F(\mathcal{E}(\rho_{\rm{in}}) ,\rho_t)\,.
\end{equation}
This is the task of finding the best state to input into a quantum channel such that the corresponding output state is the closest in fidelity to a fixed target state~\cite{zhang2022quantum}. This is the same problem as finding the optimal state of the fidelity of asymmetry.  We can use the fixed-point algorithm to find the optimal input state for the quantum error precompensation problem in the case of group averaging channels (e.g., the dephasing channel).

\item \emph{Maximum guessing probability}. Let us consider a state $\rho_A$ and its purification $\rho_{AE}$. Moreover, let us assume that Alice holds the system $A$ and Eve holds $E$. The fidelity of coherence of $\rho_{A}$ in a fixed basis gives Eve's maximum guessing probability about Alice's outcome after she measures it in the same fixed basis~\cite[Section B]{coles2012unification}. 
\end{itemize}

\section{Numerical comparison between the SDP solver and the fixed-point algorithm}
\label{numerical comparison}
In this section, we present numerical evidence of the efficiency of the fixed-point algorithm, highlighting the motivation behind our interest in it. Specifically, we compare our fixed-point iterative algorithm with SDPT3, a commercial solver available in MATLAB's CVX package~\cite{grant2014cvx}, which employs the dual interior-point method.

We study and compare the performance of the SDP and the iterative algorithm for the max-conditional entropy and the fidelity of coherence (see Section~\ref{quantum info} for a more detailed discussion). Explicitly, we compute the values 
\begin{align}
\label{optimization problems}
\max_{\sigma_B \in \mathcal{S}(B)}F(\rho_{AB},I_A \otimes \sigma_B)^\frac{1}{2}
\quad \textnormal{and} \quad
\max_{\sigma \in \mathcal{S}(A)}F(\rho, \text{diag}(\sigma))^\frac{1}{2},
\end{align}
respectively. Here, $\text{diag}(\sigma)$ is the diagonal matrix containing the elements of the diagonal of $\sigma$ on the main diagonal.

Theorem~\ref{theorem convergence} provides a theoretical guarantee on the iteration complexity for the convergence of the fixed-point algorithm. Notably, it yields dimension-independent bounds for specific problems of interest, such as the coherence fidelity (corresponding to the second optimization problem in Eq.~\eqref{optimization problems}). However, for other cases—including the first problem in Eq.~\eqref{optimization problems}—the constant $\xi$ can exhibit poor scaling with the dimensions in the worst case, ultimately resulting in a dimension-dependent convergence.
In this section, our focus is solely on numerical evidence, so we do not rely on the theoretical result from Theorem~\ref{theorem convergence}. Instead, we run the algorithm and terminate it once the iterate is sufficiently close to the solution within a predefined error margin. Given the strong convexity of the objective function, we utilize the PL inequality derived in equation~\eqref{second inequality} to certify convergence. Specifically, at each iteration, we compute the norm of the gradient after applying the group averaging operation, which provides an upper bound on the distance to the optimal value.

To formulate the problem as a semidefinite program (SDP), we leverage the relationship between the Bures distance and the fidelity discussed in section~\ref{quantum info} and employ the SDP expression of the square root fidelity from~\cite{watrous2012simpler} for two positive semidefinite operators $R,S \in \mathcal{P}(A)$
\begin{align}
\sqrt{F}(R,S) = \max_{X \in \mathcal{L}(A)}
\left\{
\begin{array}{ll}
 & \frac{1}{2}\Trm[X] + \frac{1}{2}\Trm[X^\dagger]: \\
 & 
\begin{pmatrix}
R & X \\
X^\dagger & S
\end{pmatrix} \geq 0 
\end{array}
\right\}\,.
\end{align}
Here, $\mathcal{L}(A)$ is the set of all linear operators. 
Moreover, we further simplify it using representation-theoretic tools, specifically Schur's lemma, as outlined in~\cite[Section 5]{laborde2021testing}.

We considered bipartite quantum systems with subsystem dimensions ranging from 2 to 7 for the max-conditional entropy, and quantum states of dimensions ranging from 4 to 36 for the fidelity of coherence. For each value of the dimension, we test both algorithms on the same 100 randomly generated quantum states. The SDP solver returns a value with precision $10^{-8}$. To match the same precision, we set to $10^{-9}$ the accuracy of the fixed-point algorithm. This follows from the arguments of Section~\ref{quantum info} that relate the fidelity and the Bures distance optimizations and from the fact that the maximum fidelity to the set of invariant states is lower bounded by the inverse of the dimension, and we consider dimensions smaller than $10^{2}$.  The code is available on GitHub.\footnote{\href{https://github.com/GreaterFool25/MaxFidelitySymmetric}{https://github.com/GreaterFool25/MaxFidelitySymmetric}} We run both algorithms on a personal laptop with 8 GB RAM and a 2.42 GHz CPU.  We compare the runtimes in Fig.~\ref{fig:runtime}.

\begin{figure}
  \centering
  \begin{subfigure}{0.5\textwidth}
    \centering
    \includegraphics[width=\textwidth]{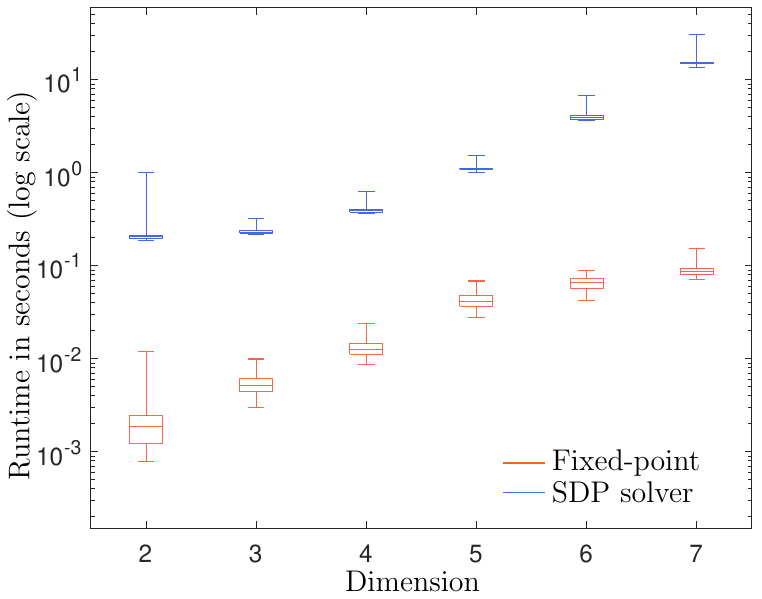}
    \caption{Max-conditional entropy}
    \label{fig:sub_a}
  \end{subfigure}
  \hskip-3pt
  \begin{subfigure}{0.5\textwidth}
    \centering
    \includegraphics[width=\textwidth]{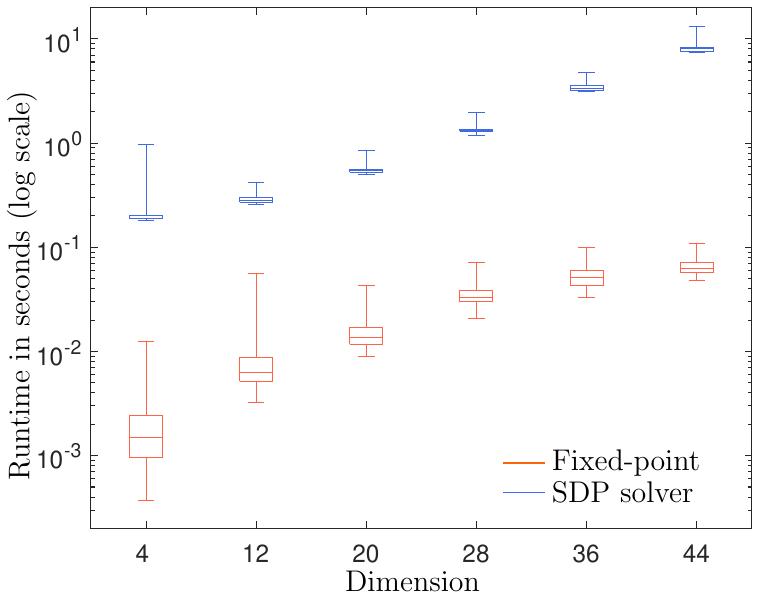}
    \caption{Fidelity of coherence}
    \label{fig:sub_b}
  \end{subfigure}
  \caption{We compare the runtime of the fixed-point iterative algorithm (red) and the SDP solver (blue). We consider two cases, namely the max-conditional entropy (a) and fidelity of coherence (b). We defined these two functions in equation~\eqref{optimization problems}. We check the performance of 100 randomly generated quantum states for several values of the dimension. We use a logarithmic scale on the y-axis, and we represent the data using a Whisker chart. The colored areas represent the interquartile regions, while the vertical lines extend to the minimum and maximum values. We considered bipartite quantum systems with subsystem dimensions ranging from 2 to 7 for the max-conditional entropy case, and quantum states of dimensions ranging from 4 to 36 for the fidelity of coherence. We run the iterative algorithm for enough iterations so it finds a state whose value is at least $10^{-5}$ close to the one returned by the SDP.}
  \label{fig:runtime}
\end{figure}

We noticed that the SDPT3 solver sometimes fails to converge within $10^{-8}$ and returns an inaccurate solution. However, the fixed-point algorithm always converges. We decided to discard these samples.
Fig.~\ref{fig:runtime} shows that the fixed-point iterative algorithm outperforms the SDP by several orders of magnitude. Moreover, we observe that the SDP algorithm requires significantly more memory than the fixed-point one.
As further evidence of the performance of our algorithm, we consider subsystem dimension $d=12$ in the max-conditional entropy. The SDP takes around twenty minutes while the fixed-point algorithm takes less than a second. We could not run the SDP for larger system sizes due to insufficient memory.

\subsection{Other methods}
\label{Other methods}
While other algorithms could be considered, we find that they perform even worse than the SDP. Below, we discuss several of these alternatives. We summarize the findings of this section, along with the relevant references for these methods, in Table~\ref{tab:alg-comparison}.

\setlength{\tabcolsep}{10pt}
\begin{table}[ht]
\centering
\begin{tabular}{lcc}
\toprule
\textbf{Algorithm} & \textbf{Theoretical Convergence} & \textbf{Practical performance} \\
\midrule
PGD~\cite{bhatia2018strong,altschuler2021averaging_2} & Dimension-independent~\cite{bhatia2018strong,altschuler2021averaging_2} & Slow \\
RGD (non-unit)~\cite{altschuler2021averaging_2} & Dimension-independent~\cite{altschuler2021averaging_2} & Slow \\
FP~\cite{alvarez2016fixed} & Dimension-dependent~\cite{chewi2020gradient} & Fast \\
\bottomrule
\end{tabular}
\caption{Comparison of projected gradient descent (PGD), Riemannian gradient descent with non-unit step size (RGD, non-unit), and the fixed-point algorithm (FP). The table summarizes both their theoretical convergence guarantees and practical performance. We cite both the original authors who introduced each algorithm and the works that established the corresponding convergence guarantees.  The “Theoretical Convergence” column in the table distinguishes between dimension-independent and dimension-dependent theoretical guarantees for the state-of-the-art algorithms. Observations regarding practical performance are based on the numerical analysis presented in Section~\ref{Other methods}. In particular, we evaluated performance using step sizes that come with theoretical guarantees on the number of iterations required for convergence.}
\label{tab:alg-comparison}
\end{table}

We conducted experiments using Riemannian descent with the specific step size proposed in~\cite{altschuler2021averaging_2} for the Bures-Wasserstein barycenter problem.  The initial point is the same as our fixed-point algorithm in Section~\ref{remark barycenter}. The step size $\eta=\lambda_{\min}/(2\lambda_{\max})$ assuming that all $X_j \in [\lambda_{\min}I,\lambda_{max}I]$ and the iteration is
\begin{align}
    &S_n=Y_n S_{n-1} Y_n \,, \quad  \text{where} \quad Y_n=(1-\eta)I+\eta \sum\nolimits_{j=1}^m \omega_j X_j^\frac{1}{2}\Big(X_j^{-\frac{1}{2}}S_{n-1}^{-1}X_j^{-\frac{1}{2}}\Big)^\frac{1}{2} X_j^\frac{1}{2} \,.
\end{align}

However, this method proved to be slow, largely due to the step size being tied to the condition number of the input matrix. As a result, the algorithm's performance deteriorates in higher dimensions, where poor conditioning becomes more pronounced. While the method theoretically ensures dimension-independent convergence that is exponentially fast in the number of iterations, the step size is limited by this dependency. Specifically, the step size shrinks as the condition number grows. In high-dimensional scenarios where the minimum eigenvalues are small, the condition number can become excessively large, drastically reducing the algorithm's efficiency. As a result, the method becomes impractical for handling ill-conditioned matrices.
The authors suggest that, in practice, the step size could be set to one, in which case the algorithm simplifies to the fixed-point method analyzed in this work. Nonetheless, a rigorous proof of this claim remains an open problem.
 We illustrate this behavior in Fig.~\ref{fig:typical behavior}.

\begin{figure}
  \centering
    \includegraphics[width=.6\textwidth]{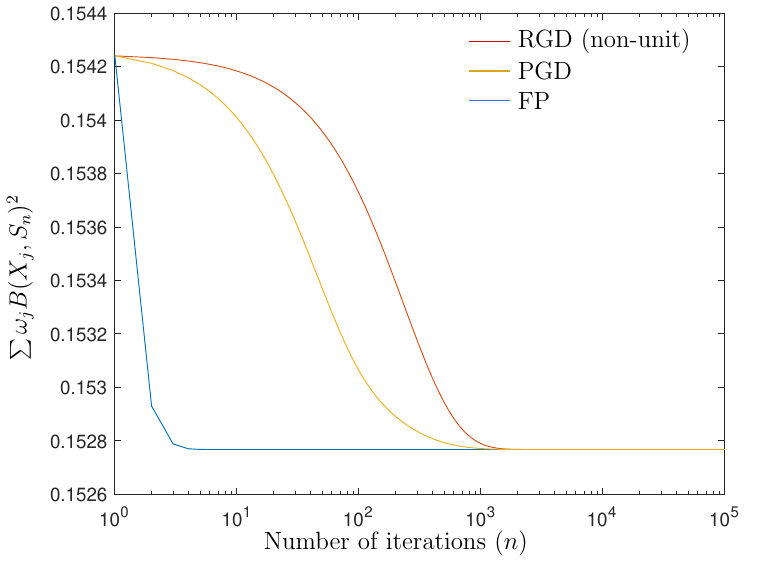}
    \caption{The plot shows the value of $\sum_j \omega_j B(X_j, S_n)^2$ as a function of the iteration number $n$, displayed on a logarithmic scale for the three algorithms discussed in the main text: the fixed-point algorithm (FP), projected gradient descent (PGD), and Riemannian gradient descent with non-unit step size (RGD (non-unit)). The example considers an average of $m = 3$ matrices in dimension $d = 4$. The plot illustrates that both RGD and PGD converge more slowly, primarily due to the use of very small step sizes in higher-dimensional settings.}
    \label{fig:typical behavior}
\end{figure}

We also compared our algorithm with the projected gradient descent, as analyzed in~\cite[Section 4]{bhatia2018strong} and~\cite[Appendix C.2]{altschuler2021averaging_2}. Here, we adopt the step size from the former work, as its larger value empirically demonstrates faster convergence than the latter. The initial point is the same as our fixed-point algorithm in Section~\ref{remark barycenter}. The step size $\eta =4\lambda_{\min}^\frac{3}{2}/\lambda_{\max}^\frac{1}{2}$ assuming that all $X_j \in [\lambda_{\min}I,\lambda_{max}I]$ and the iteration is
\begin{align}
    &S_n=\left[S_{n-1}-\eta\left(I-\sum\nolimits_{j=1}^m \omega_j X_j^\frac{1}{2}\Big(X_j^{-\frac{1}{2}}S_{n-1}^{-1}X_j^{-\frac{1}{2}}\Big)^\frac{1}{2} X_j^\frac{1}{2})\right)\right]_+ \,.
\end{align}
Here, $[\cdot]_+$ denotes the projection onto the compact set of matrices $[\lambda_{\min}I, \lambda_{\max}I]$. This projection is computed by truncating all eigenvalues of the matrix: any eigenvalue below $\lambda_{\min}$ is set to $\lambda_{\min}$, and any eigenvalue above $\lambda_{\max}$ is set to $\lambda_{\max}$. 
Projected gradient descent performs comparably to SDP solvers in low-dimensional settings, but its effectiveness diminishes in higher dimensions. 
This decline is attributed to the step size being tied to the inverse of the smoothness parameter. Specifically, when the Bures distance is involved, the step size shrinks as the condition number grows. Thus, we observe behavior similar to that of Riemannian descent with a non-unit step size, as shown in Fig.~\ref{fig:typical behavior}.
In Table~\ref{tab:algo_results} we compare the runtimes in seconds of all  algorithms discussed above.

\begin{table}[h]
\centering
\setlength{\tabcolsep}{12pt} 
\renewcommand{\arraystretch}{1.5} 
\begin{tabular}{lcc}
\toprule
\textbf{Algorithm} & $d=2$ & $d=4$ \\
\midrule
FP    & 0.007 & 0.02 \\
SDP   & 0.44 & 1.04 \\
RGD (non-unit)  & 10.89 & 55.45 \\
PGD  & 396.57 & 530.81 \\
\bottomrule
\end{tabular}
\caption{The table compares the runtime in seconds of four algorithms for the Bures-Wasserstein barycenter problem: Fixed-Point (FP), Riemannian Gradient Descent with non-unit step size (RGD (non-unit)), SDPT3 solver (SDP), and Projected Gradient Descent (PGD). The step sizes are selected according to the criteria outlined in Section~\ref{Other methods}. We randomly generate $20$ sets of matrices and weights with $m=3$, using matrix dimensions $d=2$ and $d=4$, and display the runtime average. The FP algorithm demonstrates superior performance across all test cases. We terminate the algorithms once a precision of $10^{-9}$ is reached, verifying this condition using the PL inequality as described at the beginning of  Section~\ref{numerical comparison}.}
\label{tab:algo_results}
\end{table}

\medskip
Finally, we mention that our problem could be solved using mirror descent algorithms, such as those presented in~\cite{li2019convergence,you2022minimizing}, which apply to very general convex optimization problems. However, these studies do not provide convergence guarantees in terms of the number of iterations. The work in~\cite{kum2022gpm} proposes a projected gradient descent method with Armijo rule step size selection, though it only establishes asymptotic convergence guarantees.

\section{Conclusion and further directions}
\label{conclusion}
We introduced a new fixed-point algorithm that generalizes the method originally proposed for Bures-Wasserstein barycenters to a broader class of problems, encompassing important applications in quantum information theory. We provided rigorous convergence guarantees for the fixed-point algorithm and demonstrated its practical efficiency through numerical experiments. While our results establish dimension-independent convergence for certain problems of interest, the proof remains dimension-dependent for other relevant cases. In our analysis, this limitation ultimately arises from the fact that the square function is not operator monotone unless the operators commute. As a result, we had to resort to a more involved analysis based on the determinant, as discussed in Section~\ref{Uniform lower bounds}, which leads to poor dimensional scaling and ultimately a dimension-dependent convergence guarantee in the general case. However, we believe this issue could potentially be overcome, and we leave this as an open problem for future work. 

As mentioned in the introduction, our fixed-point algorithm can also be interpreted as a Riemannian descent on the Bures Wasserstein manifold with unit step size. A promising direction for future work is to explore whether this geometric viewpoint can provide deeper insights and tighter characterizations. Moreover, it may allow for extensions to stochastic variants of the problem. Of particular interest is whether this approach can yield theoretical dimension-independent guarantees, similar to those established in~\cite{altschuler2021averaging_2} for Riemannian descent with non-unit step sizes.

Our numerical simulations suggest that similar algorithms could be used to solve analogous problems for different families of quantum R\'enyi divergences~\cite{tomamichel2015quantum}, double optimization problems as the R\'enyi mutual information in~\cite{li2022operational}, and quantum error precompensation for more general quantum channels~\cite{zhang2022quantum} (e.g., depolarizing channel).
Finally, a broader range of constraints could be incorporated into the analysis by projecting each iteration onto the set of feasible solutions after the action of the fixed-point map. We have numerical evidence that such fixed-point algorithms perform well in practice; however, it appears that a theoretical investigation of convergence guarantees for such problems requires additional techniques, and we leave this as an open question.

\section{Acknowledgment}
The authors would like to thank Hao-Chung Cheng, Mario Berta, Richard Kueng, Ian George, Erkka Haapasalo, Afham, and Antonios Varvitsiotis for discussions. In particular, we thank Christoph Hirche for suggestions at an earlier stage of the project. SB, RR, and MT were supported by the National Research Foundation, Singapore, and A*STAR under its CQT Bridging Grant. In the latter stages of the work, SB was supported by Duke University's  ECE Departmental fellowship. MT is also supported by the National Research Foundation, Singapore, and A*STAR under its Quantum Engineering Programme (NRF2021-QEP2-02-P05). This project is supported by the NRF Investigatorship award (NRF-NRFI10-2024-0006).

\newpage

\bibliographystyle{ultimate}
\bibliography{my}

\newpage

\appendix
\section*{Appendices}
\addcontentsline{toc}{section}{Appendices}
\renewcommand{\thesubsection}{\Alph{subsection}}

\subsection{Optimal choice of the initial point}
\label{initial point}
In this appendix, we provide a detailed justification for the selection of the specific initial
point in our algorithm. Specifically, we establish that the initial point of our fixed-point
algorithm corresponds to the solution in the case where the input state commutes with
the proposed initial point. Within the framework of the Bures-Wasserstein barycenter problem, this
scenario coincides with the case in which all matrices mutually commute. Moreover, our
numerical observations suggest that this initial point frequently offers a close approximation to the optimizer of the original problem, even in the general non-commutative setting.
Additionally, we demonstrate that the initial point represents the solution to an alternative quantity, which, in certain special cases—such as the commuting case—reduces to the
Bures distance.

Let $R, S \in \mathcal{P}(A)$. We define
\begin{equation}
    \widebar{B}\hskip0.7pt (R,S) = \left(\Tr(R)+\Tr(S)-2 \Tr(R^\frac{1}{2}S^\frac{1}{2})\right)^\frac{1}{2} \,.
\end{equation}
This quantity is related to the Petz R\'enyi relative entropy of order $1/2$ (see e.g.,~\cite{tomamichel2015quantum} for a review) and is equal to the Bures distance if the two matrices commute.
In the remaining part of this appendix, we will be concerned with the minimization problem
\begin{equation}
\label{third problem}
    \argmin_{S \in \mathcal{P}_G(A)} \,\widebar{B}\hskip0.7pt(R,S)^2 \,.
\end{equation}
 The problem~\eqref{third problem} is similar to the one in~\eqref{first problem} where $B$ is replaced by $\widebar{B}$. However, the above problem admits a closed-form solution; the optimal positive semidefinite matrix is given by $(\mathcal{E}(R^\frac{1}{2}))^2$.

  \begin{lemma}
     Let $R$ be a positive definite matrix. Then, $(\mathcal{E}(R^\frac{1}{2}))^2 \in \argmin_{S \in \mathcal{P}_G(A)} \,\widebar{B}\hskip0.7pt(R,S)^2$. Moreover, if $[R,(\mathcal{E}(R^\frac{1}{2}))^2] =0$, we have that $(\mathcal{E}(R^\frac{1}{2}))^2 \in \argmin_{S \in \mathcal{P}_G(A)} \,B(R,S)^2$.
 \end{lemma}
 \begin{proof}
 In the following, we use the notation $\,\widebar{B}(R,S) :=\widebar{B}_R(S)$. 
 The proof follows an argument similar to that used in the derivation of the fixed-point condition in Lemma~\ref{fixed-point equation}, employing the derivative and the integral representation of the square root function as in~\cite[Theorem 4]{rubboli2022new}. In particular, we find that for a positive matrix $R$, the optimizer $T$ must be full-rank. Moreover, for a positive matrix $A>0$, the integral representation of the square root function is
\begin{equation}
    A^\frac{1}{2}= \frac{1}{\pi}\int_0^\infty A (A+t I)^{-1} t^\frac{1}{2} \d t \,.
\end{equation}
We can then utilize the well-known formula for the derivative of the inverse of a function 
$\sigma(x)$ of $x$: $\d \sigma(x)^{-1}/\d x = - \sigma(x)^{-1} (\d \sigma(x)/\d x)  \sigma(x)^{-1} $ to obtain
\begin{align} 
\frac{\d}{\d x} (T + x Z)(T+xZ+tI)^{-1} =  t (T+xZ + t)^{-1}Z(T+xZ + tI)^{-1} \,.
\end{align} 
Hence, the derivative in $T$ along $S$ is
\begin{align}
    D\hskip1.5pt \widebar{B}_R(T)^2(Z) &= \Trm\left[\bigg(I-\frac{2}{\pi}\int_0^\infty(T+tI)^{-1}R^\frac{1}{2}(T+tI)^{-1}t^\frac{1}{2}\d t \bigg) Z\right] \\
\label{Petz minimum}
    &= \Trm\left[\bigg(I-\frac{2}{\pi}\int_0^\infty(T+tI)^{-1}\mathcal{E}\big(R^\frac{1}{2}\big)(T+tI)^{-1}t^\frac{1}{2}\d t \bigg) Z\right] \,,
\end{align}
and must be equal to zero for any Hermitian and invariant $Z$ as the optimizer $T$ lies in the interior of the set of positive definite matrices. In the second line we used that $Z = \mathcal{E}(Z)$ and that $(T+t)^{-1}$ commutes with each element of the averaging $U_g$ since both $Z$ and $T$ are invariant. This allows us to move the averaging of $\mathcal{E}(Z)$ inside the integral and apply it to $R^\frac{1}{2}$.

We then need to verify that the proposed solution $T=(\mathcal{E}(R^\frac{1}{2}))^2$ satisfies the conditions for the minimum. Using the ansatz, equality~\eqref{Petz minimum} becomes $\Trm\big[\big(I -I)X\big] = 0$ for any Hermitian $X$. The latter inequality is always satisfied for an $X$. This proves that the ansatz is the optimizer.

To prove the second part of the statement, we note that in the specific case that the input matrix commutes with $(\mathcal{E}(R^\frac{1}{2}))^2$, i.e., $[R,(\mathcal{E}(R^\frac{1}{2}))^2] = 0$, the necessary and sufficient conditions for the optimizer of the problem~\eqref{third problem} and the original one for the Bures distance~\eqref{first problem} become the same (see proof of Lemma~\ref{fixed-point equation}). Hence, in this situation, the solutions to the two problems coincide. In other words, whenever $[R,\mathcal{E}(R^\frac{1}{2})] = 0$, the positive semidefinite matrix $(\mathcal{E}(R^\frac{1}{2}))^2$ is the solution of the original Bures problem~\eqref{first problem}.
\end{proof}
As we demonstrate in Section~\ref{quantum info}, the matrix projection problem aligns with several key quantities in quantum information theory. The result presented above generalizes and recovers established findings for the Petz conditional entropies of order \(1/2\)~\cite{sharma2013fundamental,tomamichel2014relating}, as well as the Petz coherence measures of order \(1/2\)~\cite{chitambar2016comparison}.
In the specific case of the Bures-Wasserstein barycenter problem, where all matrices commute with one another, the commutation relation holds. Consequently, the solution aligns with the initial point of our algorithm and corresponds to the power means, as previously noted in~\cite{bhatia2019bures}.

\subsection{Continuity bound for non-full-rank input positive semidefinite matrices}
\label{continuity bound}
If the input matrix is not of full rank, we cannot apply Theorem~\ref{theorem convergence}. However, we could consider full-rank matrices that are arbitrarily close to it for which the theorem still holds. Indeed, for these matrices, we can give strong convergence guarantees. In the following, we derive a continuity bound that bounds the error we make using this approximation.  Explicitly, we can run the algorithm on a slightly perturbed matrix $R \rightarrow \tilde{R} = (1-\varepsilon) R +\varepsilon \Tr[R] \frac{I}{d}$ and approximate the true value arbitrarily well. Here, $d$ is the dimension of $R$. We note that for a small value of $\varepsilon$ the condition number gets very large. This could make the convergence to the solution very slow. However, in practice, as discussed at the end of Section~\ref{combining}, we find that even for small values of $\varepsilon$ the algorithm converges to the solution very fast. It is easy to check that the trace distance can be bounded as
\begin{equation}
    \frac{1}{2}\|\tilde{R}-R\|_1 =  \frac{1}{2}\|\varepsilon(\Tr[R]\pi_d-R)\|_1 \leq \varepsilon \Tr[R] \,.
\end{equation}
We then have
\begin{proposition} \label{Continuity monotones}
Let $R \in \mathcal{P}(A)$. Then for any $\tilde{R} \in \mathcal{P}(A)$ such that $\|\tilde{R}-R\|_1 \leq 2\varepsilon \Tr[R]$ and $\Tr[R]=\textup{Tr}[\tilde{R}]$ we have  
\begin{equation}
\label{continuity monotones_eq}
\left|\min_{S \in \mathcal{P}_G(A)} B(\tilde{R},S)^2 -\min_{S \in \mathcal{P}_G(A)} B(R,S)^2 \right| \leq 2 \varepsilon^\frac{1}{2} \Tr[R] \,.
\end{equation}
\end{proposition}

\begin{proof}
We follow a similar proof to the one given in~\cite[Corollary 4]{rubboli2022fundamental}. We first note that
\begin{equation}
    \min_{S \in \mathcal{P}_G(A)} B_{\tilde{R}}(S)^2 -\min_{S \in \mathcal{P}_G(A)} B_R(S)^2 \leq B_{\tilde{R}}(T_R)^2 - B_R(T_R)^2 = 2(F(R,T_R)^\frac{1}{2}-F(\tilde{R},T_R)^\frac{1}{2}) \,,
\end{equation}
where we denoted with $T_R$ the optimizer of $B_R(S)^2$.
Since the upper bound in Proposition~\ref{Continuity monotones} is a function of $\varepsilon$, we can assume the worst-case scenario $\|\tilde{R}-R\|_1 = 2\varepsilon \Tr[R]$. We set $R - \tilde{R} = P-Q$ where $P$ and $Q$ are the positive and negative parts, respectively.  We then have
\begin{align}
\label{C1}
&0 = \Trm [R - \tilde{R}] = \Tr[P-Q] \\
\label{C2}
& 2\varepsilon \Tr[R]=  \Trm[|R - \tilde{R}|] = \Tr[P]+\Tr[Q] = 2\Tr[P] \,,
\end{align}
where in the last equality of~\eqref{C2} we used~\eqref{C1}. 
It follows that $2\Tr[P] = 2\varepsilon \Tr[R]$.  We define the normalized positive semidefinite matrix $\hat{P} := P/\Tr[P]$. We then use that $ R \leq R + Q = \tilde{R} + P = \tilde{R} + \Tr[P] \hat{P} = \tilde{R} + \varepsilon \Tr[R] \hat{P}$ and we obtain
\begin{align}
&R \leq \tilde{R} + \varepsilon \Tr[R] \hat{P} \\
\label{Impl1}
\implies \quad & \Trm[(T_R^\frac{1}{2} R T_R^\frac{1}{2})^\frac{1}{2}] \leq \Trm[(T_R^\frac{1}{2} (\tilde{R}+\varepsilon \Tr[R] \hat{P}) T_R^\frac{1}{2})^\frac{1}{2}] \\
\label{Impl2}
\implies \quad &  \Trm[(T_R^\frac{1}{2} R T_R^\frac{1}{2})^\frac{1}{2}]
\leq \Trm[(T_R^\frac{1}{2} \tilde{R} T_R^\frac{1}{2})^\frac{1}{2}] + (\varepsilon\Tr[R])^\frac{1}{2}\Trm[(T_R^\frac{1}{2}  \hat{P} T_R^\frac{1}{2})^\frac{1}{2}]\\
\label{Impl3}
\implies \quad &  \Trm[(T_R^\frac{1}{2} R T_R^\frac{1}{2})^\frac{1}{2}]
\leq \Trm[(T_R^\frac{1}{2} \tilde{R} T_R^\frac{1}{2})^\frac{1}{2}] + \varepsilon^\frac{1}{2} \Tr[R] \,, 
\end{align}
where in~\eqref{Impl1} we used that the trace functional $M \rightarrow \Tr[f(M)]$ inherits the monotonicity from $f$~(see e.g.,~\cite{carlen2010trace}) and in~\eqref{Impl2} we used that for two positive semidefinite matrices $P$ and $Q$ and $\alpha \in (0,1)$ it holds  $\Tr[(P + Q)^\alpha] \leq \Tr[P^\alpha] + \Tr[Q^\alpha]$~\cite{bhatia1997matrix,marwah2022uniform}. The last implication~\eqref{Impl3} follows from the inequality 
\begin{align}
    \Trm[(T_R^\frac{1}{2} \hat{P} T_R^\frac{1}{2})^\frac{1}{2}] = \Trm[T_R]^\frac{1}{2} \Trm[(\hat{T}_R^\frac{1}{2} \hat{P} \hat{T}_R^\frac{1}{2})^\frac{1}{2}] \leq \Tr[R]^\frac{1}{2} \,,
\end{align} 
where we denoted $\hat{T}_R = T_R/\Tr[T_R]$. In the last inequality, we used that the fidelity of two states is upper-bounded by $1$ and that, as we show in Section~\ref{quantum info}, it holds
\begin{equation}
\Tr[T_R] = \max_{\sigma \in \mathcal{S}_G(A)}F(R,\sigma) = \Tr[R]\max_{\sigma \in \mathcal{S}_G(A)}F(\hat{R},\sigma) \leq \Tr[R] \,,
\end{equation}
 Here, we defined $\hat{R}:=R/\Tr[R]$.
Since the above relation also holds if we exchange $R$ and $\tilde{R}$, the proposition follows.
\end{proof}

\end{document}